\documentclass[twoside,runningheads]{llncs}

\usepackage[T1]{fontenc}
\usepackage{graphicx}

\usepackage{dlogic}
\usepackage{algorithm}
\usepackage{algpseudocode}   % algorithmicx
\usepackage[
  backend=biber,
  maxnames=99,
  url=false
]{biblatex}
\providecommand{\dAL}{\textsf{dARL}\xspace}

\ldef{ax:diff_var}{$x'$}{\D{x} = \vXp}
\ldef{ax:diff_const}{$c'$}{\D{c} = 0}
\ldef{ax:diff_plus}{$+'$}{\D{\tE + \tG} = \D{\tE} + \D{\tG}}
\ldef{ax:diff_mul}{$\cdot'$}{\D{\tE \cdot \tG} = \D{\tE} \cdot \tG + \tE \cdot \D{\tG}}

\ldef{ax:jacobian}{$\mathsf{J}$}{\D{\ttE} = \jacobian[\vvX][\ttE] \cdot \vvXp}[\vvXp \not\in \ttE]

\ldef{ax:k}{K}{\fBox[\pA][(\fP \Limplies \fQ)] \Limplies (\fBox[\pA][\fP] \Limplies \fBox[\pA][\fQ])}
\ldef{ax:r}{$[\rfle]$}{\pA \rfle \pB \Limplies \left(\fBox[\pB] \Limplies \fBox\right)}[\FV{\fP}\cap\BV{\pA, \pB}\subseteq \vvX]
\ldef{ax:tr}{$\leq^t_{\vvX}$}{\pA \rfle \pB \Land \pB \rfle \pC \Limplies \pA \rfle \pC}
\ldef{ax:dileq}{DI$_\preceq$}{\tE \preceq 0 \Limplies \fBox[\dap{\vvX}{\D{\tE} \leq 0}][\tE \preceq 0]}[\vvXp \not\in \tE]
\ldef{ax:dhc}{DHC}{\D{\tE} = 0 \Limplies \fBox[\dap{\vvX}{\tE = 0}][\D{\tE} = 0]}[\vvXp \not\in \tE]
\ldef{ax:dw}{DW}{\fBox[\dap][\fF]}
\ldef{ax:dc}{DC}{\fBox[\dap][\fR] \Limplies \dap \rfle \dap{\vvX}{\fR}}
\ldef{ax:dr}{DR}{\Lforall{\vvZ} \Lforall{\vvZp} \left(\dap{\vvX,\vvZ}{\fR \Land \fF} \rfle \dap\right)}[\vvZ,\vvZp \notin \fF]
\ldefm{ax:ag}{AG}{
  &\fBox[\dap{\vvX}{\fF}][(\ttE = \ttZero \Land \mdet{\jacobian[\vvXp][\ttE]}\neq 0)]  \Limplies \\
  & \quad \Lexists{\vvZ}\exists{\vvZp}  \left(
    \dap
    \rfle
    \dap{\vvX, \vvZ}{\fF \Land \vvZ = \ttG}
  \right)
}[\vvZ, \vvZp \not\in \dap, \ttG]
\ldefm{ax:dg}{DG}{
  &\fBox[\dap][\mdet{\mA} \neq 0] \Limplies \\
  &\quad \Lforall{\vvZ}\Lexists{\vvZp}\left(\dap \rfle \dap{\vvX,\vvZ}{\fF \Land  \mA \vvZp = \mB\vvZ + \ttC}\right)}[\vvZ, \vvZp \not\in \dap,\mA,\mB,\vvC]
\ldef{ax:dm}{DM}{\dap \rfle \dap{\vvX}{\fG} \Limplies \dap{\vvX}{\fF \Land \fR} \rfle \dap{\vvX}{\fG \Land \fR}}

\ldefm{ax:dp}{DP}{
  &\fBox[\dap{\vvX}{
    \exists \vvYp \exists \vvY \fF
  }][\fF]\Limplies \fBox[\dap{\vvX, \vvY}{\exists \vvYp \exists \vvY\fG}][\fG] \Limplies\\
  &\quad \dap \rfle \dap{\vvX}{\fG} \Limplies \dap{\vvX,\vvY}{\fF} \rfle[\vvX,\vvY]  \dap{\vvX,\vvY}{\fG}
}

\begin{document}
\title{A Deductive Refinement Calculus for Differential-Algebraic Programs}
\author{Jonathan Hellwig \inst{1} \and
        Long Qian \inst{2} \and
        André Platzer \inst{1}      
}
\institute{
  Karlsruhe Institute of Technology, Karlsruhe, Germany \\
  \email{\{jonathan.hellwig, platzer\}@kit.edu} \and
  Carnegie Mellon University, Pittsburgh, USA \\
\email{longq@andrew.cmu.edu}
}
\date{}
\maketitle

\begin{abstract}
This paper presents \emph{differential-algebraic refinement logic} (\dAL) with which one can deductively verify both properties and relations of \emph{differential-algebraic programs} (DAPs) that extend hybrid dynamical systems with differential-algebraic equations (DAEs).
A refinement calculus is introduced that enables the sound comparison of trajectories of differential-algebraic equations, crucially utilizing a novel trace-based semantics. This enables the incremental verification/simplification of complicated DAEs, while ensuring correctness at each step by the soundness of the calculus.
The calculus is shown to be complete for certifying index reductions of DAEs, providing trustworthy syntactic proofs of correctness at each step of the reduction.
\end{abstract}
\begin{keywords}
differential dynamic logic \and differential-algebraic equations \and index reduction
\end{keywords}

\section{Introduction}
Cyber-physical systems, like autonomous vehicles or medical devices, are governed by physical laws and control code whose interaction determines their behavior.
Differential dynamic logic (\dL) \cite{DBLP:journals/jar/Platzer08,DBLP:journals/jar/Platzer17,Platzer10,Platzer18} provides a deductive framework for reasoning about such systems, extending dynamic logic to verify correctness properties of hybrid systems.
However, \dL is built on top of \emph{ordinary differential equations} (ODEs).
Importantly, a large class of physical systems ranging from mechanical multi-body systems to electrical circuits, are naturally modeled by \emph{differential-algebraic equations}.
This calls for a deductive framework beyond \dL.
\emph{Differential-algebraic programs} (DAPs) \cite{DBLP:journals/logcom/Platzer10} extend \dL to differential-algebraic equations, but were originally restricted to a special normal form \cite{DBLP:journals/logcom/Platzer10} or for arbitrary semi-algebraic constraints under the assumption that solutions are real-analytic~\cite{DBLP:conf/cade/HellwigP25}.
This means that even a simple Coulomb friction model, expressed as the DAE $x'=v, m v' = -c |v|$, fall outside the scope of deductive verification. Its solutions are $C^1$, but not real-analytic near $v =0$.

Beyond the question of regularity, DAEs pose a second challenge: \emph{index reduction}, a technique to reveal hidden implicit constraints in DAEs before verification can proceed.
Prior work on the verification of DAEs addresses such issues, but does not support refinement-based reasoning.
Conversely, refinement calculi for hybrid systems have been developed in~\cite{DBLP:conf/lics/LoosP16,DBLP:conf/ijcar/PrebetP24}, enabling formal comparison and simplification of hybrid programs but only for ODEs.

\subsubsection{Contributions.}
This paper introduces \emph{differential-algebraic refinement logic} (\dAL), addressing the two limitations outlined above through the following contributions. First, we introduce a trace-based semantics for differential-algebraic programs grounded in $C^1$ functions. Compared with the real-analytic semantics of prior work, this semantic foundation admits a larger class of practically relevant systems, including Coulomb friction and dynamics involving $\min$ and $|\cdot|$. 
Second, we develop a sound refinement calculus for DAPs. 
This calculus enables relational reasoning between DAPs. In particular, building on this calculus we present a method to algorithmically verify index-reductions.

\section{Differential-Algebraic Refinement Logic}
This section presents the syntax and semantics of differential-algebraic refinement logic (\dAL). 
The logic combines core ideas from differential refinement logic \cite{DBLP:conf/ijcar/PrebetP24,DBLP:conf/lics/LoosP16} and differential-algebraic logic \cite{DBLP:journals/logcom/Platzer10,DBLP:conf/cade/HellwigP25}.
We recall the basic syntax and semantics of differential dynamic logic and extend them with definitions for differential-algebraic programs and refinement, which together provide the foundation for refinement principles introduced in section \ref{sec:proof_calculus}.

\subsection{Syntax.}
This subsection introduces the syntax of \dAL.
The syntax follows largely that of differential dynamic logic \cite{DBLP:journals/jar/Platzer17}, and is recalled here for completeness, with the addition of refinements and differential-algebraic programs. 
We begin with the construction of terms.
Let $\Vars$ denote the set of variables.
Each variable $x \in \Vars$ has an associated \emph{differential variable} denoted by $x'$.
Differential variables are independent of the base variables; in particular, $x'$ is not an abbreviation for the derivative of $x$, but a distinct variable.
For terms, we adopt the standard term language of differential dynamic logic, including \emph{(syntactic) differentials}, defined formally below.
\begin{definition}[Term Syntax]
Let $x \in \Vars$ be a variable and $c \in \Q$ be a rational constant.
The language of terms $\tE, \tG$ is defined by the following grammar:
\[
  \tE, \tG \mathrel{::=} c \mid \vX \mid \tE + \tG \mid \tE \cdot \tG \mid \D{\tE}
\]
\end{definition}
The term language extends real arithmetic with differentials.
Intuitively, the differential $\D{\tE}$ captures the local change of the term $\tE$ under infinitesimal changes of its variables.

Formulas and differential-algebraic programs are defined by simultaneous induction, because programs can occur in formulas and vice versa.
\begin{definition}[Formula Syntax]
Let \(\tE, \tG\) be terms, \(\vvX = (x_1, \dots, x_n)\) be a finite tuple of variables and \(\alpha, \beta\) be differential-algebraic programs.
The language of formulas $\fP, \fQ$ is defined by the following grammar:
\[
  \fP, \fQ \mathrel{::=} \tE \le \tG \mid \Lnot \fP \mid \fP \Land \fQ \mid \Lforall{x} \fP \mid \fBox \mid \blue{\pA \rfle \pB}
\]
\end{definition}

As \dAL is an extension of the first-order logic of real arithmetic, the interpretations/constructions of first-order formulas are standard. The \emph{modal formula} $\fBox$ is inherited from \dL and classical dynamic logic of discrete programs, expressing the property that $\fP$ holds after \emph{every run} of the (differential-algebraic) program $\pA$. The \emph{(partial) refinement formula} $\pA \rfle \pB$ extends the existing refinement notion and is crucial to the axiomatic framework presented in this work, expressing the property that \emph{every possible evolution} of the variables $\vvX$ attained by executing the program $\pA$ can also be achieved by somehow executing the program $\pB$. Intuitively, this implies that the possible behaviors of $\vvX$ under $\pA$ is a subset of the possible behaviors of $\vvX$ under $\pB$, thus $\pA$ is said to be a \emph{(partial) refinement} of $\pB$ with respect to $\vvX$. 
Mutual refinement of programs is denoted by $\pA \rfeq \pB$, abbreviating $\pA \rfle \pB \land \pB \rfle \pA$.

The following definition introduces the syntax of \emph{differential-algebraic programs} (DAPs).
\begin{definition}[Program Syntax]
Let $x$ be a variable, $\tE$ be a term and let $\fQ,\, \fF$ be formulas of real arithmetic.
The language of programs $\pA, \pB$ is defined by the following grammar:
\[
  \pA, \pB \mathrel{::=} \assign \mid\;  \test \mid \blue{\dap} \mid \choice \mid \compose \mid \rep
\]
\end{definition}

Programs in \dAL model hybrid systems by combining discrete dynamics with continuous evolution.
In contrast to \dL, which restricts continuous evolution of $\vvX$ to \emph{ordinary differential equations}, \dAL introduces \emph{differential-algebraic programs} (DAP) which permit evolution governed by arbitrary real-arithmetic formulas over $\vvX$ and $\vvXp$.
This generalization allows modeling of disturbances and implicit dynamics.
An important special case of differential-algebraic programs are \emph{differential-algebraic equations}, which consist of a set of algebraic equations $\tE_1(\vvX,\vvXp) = 0, \dots, \tE_n(\vvX,\vvXp) = 0$. 
Unlike ODEs, such equations are in general not explicitly solvable for $\vvXp$ and need not admit locally unique solutions.

The remaining program constructors characterize discrete program structure and control flow.
Intuitively, programs $\assign$ and $\test$ are discrete and represents \emph{discrete assignments} and \emph{tests}. Programs $\choice,\, \compose, \, \rep$ are program combinators that represent \emph{non-deterministic choice} (either one of $\pA, \pB$ can be executed), \emph{sequential composition} ($\pA$ is executed and then $\pB$ is executed) and \emph{(unbounded) repetition} ($\pA$ can be executed any finite number of times) respectively.

\subsubsection{Notation.}
 To simplify the presentation, we adopt a vector and matrix notation.
We denote vectors of variables by \(\vvX = (\vX_1,\dots,\vX_n)^\top\), vectors of terms by \(\ttE = (\tE_1,\dots,\tE_n)^\top\), and matrices of terms by \(\mA = (a_{ij})\).
The matrix-vector product \(\mA \ttE\) is an abbreviation for a term vector whose \(i\)-th component is the \(i\)-ith coordinate of the matrix product \((\mA \ttE)_i = \sum_{j=1}a_{ij}\tE_j\).
Relational comparison for vectors $\ttE, \ttG$ of the form \(\ttE \leq \ttG\) abbreviate the conjunction \(\bigwedge_{i=1}^n \tE_i \leq \tG_i\).

\subsubsection{Expressiveness of Differential-Algebraic Programs.}
  Explicit ordinary differential equations with domain constraints, as used in \dL, are subsumed as a special case of the following DAP:
  \[
    \dap{\vvX}{\vvXp = \ttF(\vvX) \land \fQ}.
  \]
  DAPs also allow a much broader class of systems, including, for example, linear differential-algebraic equations
  \[
  \dap{\vvX}{\mA \vvXp = \vvB},
  \]
  which in general cannot be reduced to ordinary differential equations when the matrix $\mA$ is not invertible. 
  More generally, DAPs admit arbitrary semi-algebraic constraints over state and differential variables, for example
  \[
   \dap{\vvX}{\ttG(\vvX,\vvXp) \ge 0 \land \ttH(\vvX,\vvXp) = 0 \land \ttR(\vvX,\vvXp) > 0}.
  \]

\subsection{Semantics.}
Prior work on DAPs defines the semantics of programs via state transition relations \cite{DBLP:journals/jar/Platzer08,DBLP:conf/cade/HellwigP25}, representing executions via initial and final state pairs.
Such semantics are insufficient in our setting, where reasoning about intermediate states is essential, rather than just their end states.
We therefore adopt a trace-based semantics, interpreting programs as sets of continuous functions from time to states.
This perspective aligns naturally with our focus of continuous dynamics in this work.

A \emph{state} \(\sO : \Vars \mapsto \R\) is a mapping from variables to real values. 
In particular, states also assign values to differential variables.
We write \(\States\) for the set of all states and \(\sO(\vX)\) for the value of variable \(\vX\) in state \(\sO\).
The semantics of terms are standard \cite{DBLP:journals/jar/Platzer08,DBLP:conf/lics/Platzer12a}.
For a real value $r \in \R$, we use the notation $\stateUpd$ to denote the state obtained by updating $\sO$ at the variable $x$ with the value $r$.

\begin{definition}
  The \emph{semantics} of a term $\tE$ is mapping $\sem{\tE} : \States \rightarrow \R$ defined inductively:
  \begin{enumerate}
    \item $\semS{\vX}{\sO} = \sO(\vX)$,
    \item $\semS{\tE + \tG}{\sO} = \semS{\tE}{\sO} + \semS{\tG}{\sO}$,
    \item $\semS{\tE \cdot \tG}{\sO} = \semS{\tE}{\sO} \cdot \semS{\tG}{\sO}$,
    \item $\semS{\D{\tE}}{\sO} = \sum_{v \in \Vars}  \sO(v') \frac{\partial}{\partial  r} \semS{\tE}{\stateUpd[\sO][v]}$.
  \end{enumerate}
\end{definition}
Definitions $(1) - (3)$ are standard, the valuation of differential term $\semS{\D{\tE}}{\sO}$ evaluates the differential of $\D{\tE}$ using the chain rule, giving $\tE' = \sum_{v \in \Vars} v' \frac{\partial \tE}{\partial v}$, and evaluating the expression at the state $\omega \in \States$ yields definition $(4)$. 
For notational convenience,  we extend the valuation componentwise to vectors and matrices of terms:
\[
  \semS{\ttE}{\sO} = (\semS{\tE_i}{\sO})_i, \quad \semS{\mA}{\sO} = (\semS{a_{ij}}{\sO})_{ij}.
\]

We now define the core semantic objects of our programs:
A \emph{trace} \(\flA\) of duration \(\flTA \geq 0\) is a mapping \(\flA :\flInt \rightarrow \States\).
We denote the set of traces by \(\flSpace = \bigcup_{T \geq 0} \States^{[0,T]}\).
The valuation of a variable vector \(\vvX\) at time \(t\) is understood pointwise as \(\flAt(\vvX) = (\flAt(\vX_1), \dots, \flAt(\vX_n))^\top\).

A key assumption in earlier works axiomatizing continuous evolutions is that such traces are real-analytic functions in the time variable $t$ \cite{DBLP:journals/jacm/PlatzerT20,DBLP:conf/cade/HellwigP25}, which makes it impossible to switch between disjuncts with different differential-algebraic equations.
This work considers a much more general collection of traces by relaxing such regularity requirements.
\begin{definition} \label{def:x-regular}
  A trace \(\flA \in \flSpace\) is called \(\vvX\)\emph{-regular}, if
  \begin{enumerate}
    \item \(t \mapsto \flAt(\vvX)\) is of class \(C^1(\flInt,\R^n)\),
    \item \(\ddt \flAt(\vvX) = \flAt(\vvXp)\) for all \(t \in \flInt\),
    \item for all \(\vY \not\in \vvX \cup \vvXp\) and all \(t \in \flInt\), we have \(\flAt(\vY) = \flA(0)(\vY)\).
  \end{enumerate}
  The set of \(\vvX\)-regular traces is denoted by \(\flRegular[\vvX]\).
\end{definition}
\begin{remark}
  The derivatives at the boundary points \(\{0,\flTA\}\) are defined as the one-sided limits from within the interval \(\flInt\), and the duration zero trace with $T_\flA = 0, \flA(0)(\vvX) = \vvX$ is understood to be $\vvX$-regular. 
\end{remark}

The first property restricts the regularity of traces to the class of $C^1$ functions.
Definition~\ref{def:x-regular} uses a closed interval, which means it excludes traces whose first derivative blows up as it approaches the boundary.
The second property requires \emph{differential consistency}:
the value of each differential variable must be the time derivative of its corresponding state variable.
The third property constrains the bound variables of DAPs: only the variables in $\vvX \cup \vvXp$ may change along the trace, while all others must remain constant.

A central feature of $\dL$ is its integration of both discrete and continuous programs, which can be combined using sequential composition. 
While sequential composition for relational semantics can be defined easily via the composition of reachability relations, a trace-based approach requires more careful treatment to ensure a well-defined (single-valued) trace is obtained rather than a general relation. 
To this end, we define the \emph{concatenation operator} \(\cdot : \flSpace \times \flSpace \rightarrow \flSpace\) such that for \(\flA, \flB \in \flSpace\), the concatenated trace \(\flConcat{\flA}{\flB}\) is defined on \([0,\flTA + \flTB]\) as
\begin{align*}
  (\flConcat{\flA}{\flB})(t) = \begin{cases}
    \flAt & \text{ for } t \in [0,\flTA), \\
    \flB(t - \flTA) & \text{ for } t \in [\flTA,\flTA + \flTB].
  \end{cases}
\end{align*}
Note that this operation is defined for any two traces.
If the transition states do not match, the trace exhibits an instantaneous jump at \(\flTA\).
Consequently, the concatenation of continuous traces is, in general, only a right-continuous trace.
For the discrete program constructs, we define the point trace at state \(\sO\), denoted \(\flPoint\). 
It is the instantaneous trace with duration \(T_{\flPoint} = 0\) and \(\flPoint(0) = \sO\).

The preceding definitions provide the necessary foundation to characterize the semantics of programs:
\begin{definition}[Program semantics]
  \label{def:prog_sem}
  The \emph{semantics} of a program \(\pA\) is a mapping \(\sem{\pA} : \States \rightarrow 2^{\flSpace}\), defined inductively:
  \begin{enumerate}
    \item \(\sem{\assign} = \sO \mapsto \{\flPoint[\stateUpd[\sO][\vX][\semS{\tE}{\sO}]]\},\)
    \item \(\sem{\test} = \sO \mapsto \{\delta_{\sO}\mid \sO \in \sem{\fQ}\},\)
    \item \(\sem{\dap} = \sO \mapsto \{\flA \in \flRegular \mid \flA(0) = \sO, \flA(\flInt) \subseteq \sem{\fF}\},\)
    \item \(
    \sem{\compose{\pA}{\pB}} = \sO \mapsto \{\flConcat{\flA}{\flB} \mid \flA \in \sem{\pA}(\sO), \;\flB \in \sem{\pB}(\flAT)
    \},\)
    \item \(\sem{\choice{\pA}{\pB}} = \sO \mapsto \sem{\pA}(\sO) \cup \sem{\pB}(\sO),\)
    \item \(\sem{\rep{\pA}}] = \sO \mapsto \bigcup_{n=0}^\infty \sem{\pA^n}(\sO)\),
  \end{enumerate}
  where \(\pA^0 \equiv{} \test{\top}\) and \(\pA^{n+1} \equiv \compose{\pA}{\pA^{n}}\).
\end{definition}

The semantics of programs is defined as a mapping from initial states to sets of traces.
For a given state $\sO \in \States$, the notation $\semPS$ denotes the set of traces originating in $\sO$.
This formulation is necessary because continuous and discrete programs treat initial states differently: continuous programs evolve from the initial state, while discrete assignment instantaneously transition to a new state.
An important property of the semantics of DAPs is that traces respect the entire initial state, including values of differential variables $\vvXp$.
This is crucial for ensuring the soundness of our proof calculus.

We say two traces \(\flA\) and \(\flB\) \emph{coincide} on a set of variables \(V\), denoted by the \emph{coincidence relation} \( \coincide[V]\), if they have the same duration \(\flTA = \flTB\) and agree on the valuation of all variables (and their differential variables) in \(V\) for all \(t \in \flInt\).
\begin{definition}[Formula semantics]
  The semantics of a formula \(\fP\) is a set of states \(\sem{\fP} \subseteq \States\) defined inductively:
  \begin{enumerate}
    \item \(\sem{\tE \leq \tG} = \{\sO \in \States \mid \semS{\tE}{\sO} \leq \semS{\tG}{\sO}\},\)
    \item \(\sem{\fP \Land \fQ} = \sem{\fP} \cap \sem{\fQ},\)
    \item \(\sem{\Lnot \fP} = \States \setminus \sem{\fP},\)
    \item \(\sem{\Lforall{\vX}\fP} =  \{\sO \in \States \mid \forall r \in \R : \stateUpd \in \sem{\fP}\},\)
    \item \(\sem{\fBox} = \{\sO \in \States \mid \forall \flA \in \semPS[\pA] : \flAT \in \sem{\fP}\},\)
    \item \(\sem{\pA \rfle \pB} = \{\sO \in \States \mid \forall \flA \in \semPS[\pA]~\exists \flB \in \semPS[\pB]: \coincide\}.\)
  \end{enumerate}
\end{definition}
The last clause defines the semantics of the (partial) refinement formula. A state $\sO$ satisfies $\pA \rfle \pB$, if for every trace $\flA$ of program $\pA$ starting in state $\sO$, there exists a corresponding trace $\flB$ in program $\pB$, such that $\flA$ and $\flB$ coincide on the tuple $\vvX$.

Although this work focuses on the continuous fragment of programs, redefining the semantics of \dL requires ensuring compatibility with its original interpretation. 
The following proposition establishes that this is indeed the case.
\begin{proposition}\label{prop: conservative_extension}
  On the syntactic fragment of \dL, the trace-based semantics and relational semantics are equivalent with respect to reachability: for any program \(\pA\) and a pair of states \((\sO, \sN) \in \States \times \States\), we have
  \((\sO,\sN) \in \semdL{\pA}\), if and only if, there exists a trace \(\flA \in \semPS\) with terminal state \(\flAT = \sN\).
\end{proposition}
This proposition establishes that \dAL is an \emph{extension} of \dL, the proof can be found in Appendix \ref{sec: appendix_conservative_extension}. In particular, all axioms for discrete programs remain valid.
\section{Axioms and Rules of \dAL} \label{sec:proof_calculus}
This section develops a sound proof calculus for \dAL. 
This calculus combines the axiomatic treatment of differential terms with additional principles for refinements.

\subsubsection{Syntactic Side Conditions.}
Syntactic side conditions are essential for the soundness of our proof calculus.
We use standard definitions \cite{DBLP:journals/jar/Platzer17} and write \(\FV{\fP}\) for the set of (syntactically) free variables of formula $\fP$ and \(\BV{\pA}\) for the set of (syntactically) bound variables of program $\pA$.
For DAPs, we set \(\BV{\dap} = \vvX \cup \vvXp\).
The following lemmas formally relate syntactically free and bound variables to our trace-based semantics.

\begin{lemma}[Coincidence]\label{lem:coincidence}
  Let \(\fP\) be a formula and \(\sO,\sN \in \States\) be states.
  If \(\coincide[V][\sO][\sN]\) for some variable set \(\FV{\fP} \subseteq V\), then 
  \(\sO \in \sem{\fP}\), if and only if, \(\sN \in \sem{\fP}.\)
\end{lemma}

\begin{lemma}[Bound effect]
  For any program \(\pA\) and trace \(\flA \in \semPS\), we have \(\sO(v) = \flAT(v)\) for all \(v \in \BV{\pA}^\complement\).
\end{lemma}

The coincidence lemma states that if two states agree on a superset of a formula's free variables, then they agree on whether the formula is satisfied.
The bound effect lemma expresses that a program execution can only modify the values of its bound variables, leaving all other variables unchanged along the trace.

\subsubsection{Differential Term Axioms.}
Differential terms are syntactically axiomatized by the differential
axioms for variables, constants, sums, and products \cite{DBLP:journals/jar/Platzer17}:
\begin{align*}
  \axtag{ax:diff_var} & \quad \lquote{ax:diff_var} \qquad \axtag{ax:diff_plus} \quad \lquote{ax:diff_plus}\\
  \axtag{ax:diff_const} & \quad \lquote{ax:diff_const} \qquad\quad \axtag{ax:diff_mul} \quad \lquote{ax:diff_mul} 
\end{align*}

Because our calculus is based on $C^1$ traces, rather than real-analytic ones, some of our reasoning principles require additional regularity assumptions.
To state these assumptions syntactically, we define syntactic Jacobians of term vectors.

\begin{definition}[Syntactic partial derivative and Jacobian]\label{def:partial_diff}
Let $\vX$ and $\vY$ be variables and $c$ a rational constant.
For terms $\tE,\tG$
the syntactic partial derivative $\sPart{(\tE)}{\vX}$ is defined by
\begin{align*}
  \sPart{(\tE+\tG)}{\vX}&\equiv\sPart{(\tE)}{\vX}+\sPart{(\tG)}{\vX} \qquad\qquad \sPart{(c)}{\vX} \equiv 0\\
  \sPart{(\tE \cdot \tG)}{\vX}&\equiv\sPart{(\tE)}{\vX} \tG + \tE \sPart{(\tG)}{\vX} \qquad\quad \sPart{(\vY)}{\vX}\equiv\begin{cases}1&y=x\\0&y\neq x\end{cases}
\end{align*}

Let $\vvX = (x_1,\dots,x_n)$ be a tuple of variables and let
$\ttE = (\tE_1,\dots,\tE_m)$ be a vector of terms.
The \emph{syntactic Jacobian} of $\ttE$ with respect to $\vvX$ is the $m \times n$ matrix \(\jacobian[\vvX][\ttE]\)  defined by:
\[
\jacobian[\vvX][\ttE]
\equiv \left[\sPart{\tE_i}{x_j}\right]_{i,j}
\]
\end{definition}
The axiom schema \lref{ax:jacobian} expresses the consistency between syntactic differentials and syntactic Jacobians. It is only valid if $\vvXp$ is not free in the term vector $\ttE$:
\[
\begin{aligned}
  \axtag{ax:jacobian} \quad & \lquote{ax:jacobian} \quad \lside{ax:jacobian}
\end{aligned}
\]

\subsubsection{Refinement Axioms.}
The following two \emph{refinement axioms} presented below are generalizations of the $[\leq]$ and $\leq_t$ axioms \cite{DBLP:conf/ijcar/PrebetP24} to the partial refinement case:
\begin{align*}
  \axtag{ax:r} & \quad \lquote{ax:r} \quad \lside{ax:r}\\
  \axtag{ax:tr} & \quad \lquote{ax:tr}
\end{align*}
The \lref{ax:r} axiom builds a bridge between box and refinement formulas. It is only valid under the stated side condition, which requires that any free variable of $\fP$ outside $\vvX$ may not be modified by either of programs.
Axiom \lref{ax:tr} states the usual transitivity property of refinement.

\subsubsection{Box Axioms.}
The following axioms express basic reasoning principles for box modalities of DAEs.
Although they follow from elementary properties of $\vvX$-regular traces, these axioms are sufficiently expressive to derive non-trivial properties of DAPs (c.f. Section \ref{sec:index_reduction}). We write $\preceq$ to denote either $<$ or $\leq$:
\begin{align*}
  \axtag{ax:dw} & \quad \lquote{ax:dw} \\
  \axtag{ax:dileq} & \quad \lquote{ax:dileq} & \lside{ax:dileq} \\
  \axtag{ax:dhc} & \quad \lquote{ax:dhc} & \lside{ax:dhc}
\end{align*}
The \emph{differential weakening} axiom \lref{ax:dw} axiomatizes the simple property that if $F$ is a constraint for the system $\dap$, then $F$ must be satisfied along all traces.
Axiom \lref{ax:dileq} is the \emph{differential induction} axiom, expressing the fact that if a differential-free term $\tE$ is initially (negative) non-positive and along all possible traces its differential is non-positive, then $\tE$ always remains (negative) non-positive.
Finally, axiom \lref{ax:dhc} is the \emph{differential hidden constraint} axiom. 
It can be viewed as the converse of the \lref{ax:dileq} axiom: if a differential-free term $e$ is identically zero on the entire time horizon, then necessarily its differential $\D{\tE}$ must also be zero.
The role of axiom \lref{ax:dhc} is discussed in Section~\ref{sec:index_reduction}.

\subsubsection{Ghost Axioms.}
The \emph{ghost axioms} provide reasoning principles for extending DAPs with additional variables while preserving refinement.
Under our $C^1$ trace-based semantics, such extensions require regularity assumptions to ensure that the added variables admit a solution on the same time interval of existence as the original trace.
To state these conditions, we first define a syntactic determinant.

\begin{definition}[Syntactic determinant]
  \label{def:det}
Let $\mA=(a_{ij})$ be an $n\times n$ matrix whose entries are terms.
We define the term $\mdet(\mA)$ by recursion on $n$:
\[
\mdet([a_{11}]) \equiv a_{11},
\qquad
\mdet(\mA) \equiv \sum_{j=1}^n (-1)^{1+j} a_{1j}\, \mdet(\mA_{1j}) \quad (n>1),
\]
where $\mA_{1j}$ is obtained from $\mA$ by deleting row $1$ and column $j$.
\end{definition}
Using this notation, we state the two extension principles:
\begin{align*}
  \hspace{-2cm}
  \axtag{ax:ag} & \quad \lquote{ax:ag}
                & \lside{ax:ag}\\
  \axtag{ax:dg} & \quad \lquote{ax:dg} & {\hspace{-2cm}\lside{ax:dg}}
\end{align*}
Axiom \lref{ax:ag} is the \emph{algebraic ghost} axiom, axiomatizing the fact that if the Jacobian associated to an algebraic equation $\ttE = \ttZero$ is invertible, then the system is actually explicit and the differentials $\vvXp$ can be solved \emph{explicitly} from the equation $\ttE = \ttZero$. This allows for the introduction of new variables without losing regularity. In particular, one can introduce constrained variables $\vvZ = \vvXp$ (using $\ttG \equiv \vvXp$) and manipulate $\vvZ$ as a differential-free variable.
Axiom \lref{ax:dg} (\emph{differential ghost}) states that (differential) variables governed by linear DAEs can be soundly added to the system, \emph{provided that} $\mA$ is invertible along all traces of the original DAE $\dap{\vvX}$.
In that case, the linear DAE is equivalent to a linear ODE obtained by solving for the differential variables $\vvZp$ via  $\mA^{-1}$. 

The following axioms provide reasoning principles to analyze subsystems of DAEs. 
\begin{align*}
  \axtag{ax:dc} & \quad \lquote{ax:dc}\\
  \axtag{ax:dr} & \quad \lquote{ax:dr} & \lside{ax:dr}\\
  \axtag{ax:dm} & \quad \lquote{ax:dm} \\
  \axtag{ax:dp} & \quad \lquote{ax:dp}
\end{align*}
Axiom \lref{ax:dc} is the \emph{differential cut} axiom, expressing that if a formula $R$ is always satisfied along a DAE $\dap{\vvX}$, then $R$ can be soundly added as a constraint. 
Axiom \lref{ax:dr} is the \emph{differential refinement} axiom, axiomatizing the fact that additional constraints only restrict the system.
Intuitively, the \emph{differential monotonicity} axiom \lref{ax:dm} expresses a monotonicity property of differential programs.
If $\dap \rfle \dap{\vvX}{\fG}$ holds, then imposing an additional constraint $R$ on both systems preserves the refinement relation.
The \lref{ax:dp} Axiom (\emph{differential projection}) states that if $\fF$ and $\fG$ are semantically independent of $\vvYp$ and $\vvY$ along all traces, then it is sound to prove the refinement for the full system $(\vvX,\vvY)$ by proving it on the $\vvX$-subsystem.

\begin{theorem}[Soundness]
    \label{thm: sound axiomatization}
  The \dAL proof calculus is sound. 
\end{theorem}
\begin{proof}
  See Appendix \ref{sec:appendix_soundness_proof}.
\end{proof}

\subsubsection{Derived Rules.}
\ldef{rl:dhc}{dHC}{Formula}
\ldef{rl:di}{dI}{Formula}
\ldef{rl:da}{dA}{Formula}
\ldef{rl:dw}{dW}{Formula}
\ldef{rl:g}{G}{Formula}
\ldef{rl:real}{$\R$}{Formula}
\ldef{rl:cut}{cut}{Formula}
\ldef{rl:unfold}{unfold}{Formula}
\ldef{rl:rtr}{rTR}{Formula}

An important feature of axiomatic, logical approaches to the verification of dynamical systems is the ability to \emph{derive} sophisticated proof rules from a small set of sound axioms by clever deductions. Such derived axioms are then guaranteed to be sound by the soundness of the base axioms (Theorem \ref{thm: sound axiomatization}). The following theorem presents the list of derived axioms established, allowing for a symbolic, deductive treatment of index reduction in Section \ref{sec:index_reduction}.

\begin{corollary}[Derived Rules]
  \label{thm:derived_axioms}
Let $\ttE$ be a differential-free term. The following are derived rules of \dAL:

\begin{center}
\begin{minipage}{0.35\textwidth}
\centering
\begin{sequentproof}[align]
  \ax{\fF}{\fP}
  \un[rl:dw]{\Gamma}{\fBox[\dap],\Delta}
\end{sequentproof}
\end{minipage}
\hfill
\begin{minipage}{0.55\textwidth}
\centering
\begin{sequentproof}
  \ax{\fF}{\fR}
  \ax{\Gamma}{\dap{\vvX}{\fR} \rfle \dap{\vvX}{\fG}, \Delta}
  \bi[rl:da]{\Gamma}{\dap \rfle \dap{\vvX}{\fG}, \Delta}
\end{sequentproof}
\end{minipage}
\end{center}

\begin{sequentproof}[align]
  \ax{\Gamma}{\ttE \preceq \ttZero, \Delta}
  \ax{\Gamma}{\fBox[\dap][\D{\ttE} \leq \ttZero], \Delta}
  \bi[rl:di]{\Gamma}{\dap \rfeq \dap{\vvX}{\fF \Land \ttE \preceq \ttZero}, \Delta}
\end{sequentproof}

\begin{sequentproof}[align]
  \ax{\Gamma}{\D{\ttE} = \ttZero, \Delta}
  \ax{\Gamma}{\fBox[\dap][\ttE = \ttZero], \Delta}
  \bi[rl:dhc]{\Gamma}{\dap \rfeq \dap{\vvX}{\fF \Land \D{\ttE} = \ttZero}, \Delta}
\end{sequentproof}
\end{corollary}

\begin{proof}
  See Appendix \ref{sec:appendex_derived_rules}.
\end{proof}

\begin{itemize}
  \item Proof rule \lref{rl:dw} is the derived version of axiom \lref{ax:dw}.
  \item Proof rule \lref{rl:da} is the \emph{differential advance} rule, where a constraint $\fF$ can be replaced by $\fR$ provided that $\fR$ is more permissive.
  \item Proof rule \lref{rl:di} is the \emph{differential induction} rule for DAPs, the derived version of axiom \lref{ax:dileq}.
  \item Proof rule \lref{rl:dhc} is the \emph{differential hidden constraint} rule for equality constraints of the form $\ttE = \ttZero$, the derived version of \lref{ax:dhc}. 
\end{itemize}

The following example shows how $\dAL$ can be used for safety verification via refinements.

\begin{example}
  \label{ex:safety}
  Consider the system $\{x, y : x' = -y \land xy = 1\}$ with variables $\vvX = (x, y)$ and suppose we wish to prove the property that starting from $x = y = 1, x' = -1, y' = 1$, the positivity condition $y > 0$ always holds. That is, the goal is to prove the following sequent
  \[\sequent{x = y = 1, x' = -1, y' = 1}{\boxx{\{x, y : x' = -y \land xy = 1\}}{y > 0}}\]
  First apply axiom \lref{ax:dg} with $\mA = \mI$ being the identity matrix to soundly introduce the differential ghost $z' = -y^2z/2$ with initial condition $z^2y = 1$ (sound as $y = 1$ initially). This application yields the refinement relation
  \[\{x, y : x' = -y \land xy = 1\} \rfle[\vvX] \{x, y, z : x' = -y \land xy = 1 \land z' = -y^2z/2\}\]
  As the desired safety condition is $y > 0$, an application of the refinement axiom \lref{ax:r} reduces the problem to proving
  {\footnotesize
  \[\sequent{x = y = 1, x' = -1, y' = 1, z^2y = 1}{\boxx{\{x, y, z : x' = -y \land xy = 1 \land z' = -y^2z/2\}}{y > 0}}\]
  }
  Further note that it suffices to prove the safety property of $z^2y = 1$, as it then follows by arithmetic reasoning that $y > 0$. To establish this, an application \footnote{Technically we use the equality version, which is equivalent to applying \lref{ax:dileq} twice with different signs.} of axiom \lref{ax:dileq} reduces the proof obligation to establishing the validity of $(z^2y)' = 0$ along the DAE. By the syntactic definition of differentials, this evaluates to 
  \[(z^2y)' = (z^2)'y + z^2(y)' = 2z'zy + z^2y' = -y^3z^2 + z^2y'\]
  By axioms \lref{ax:dw} and rule \lref{rl:da}, we may apply arithmetic reasoning and conclude that since $xy = 1$ holds, $x \neq 0$ and $(z^2y)' = 0 \leftrightarrow x \cdot (z^2y)' = 0$. The latter expression evaluates to (using $xy = 1$)
  \[x \cdot (-y^3z^2 + z^2y') = -xy^3z^2 + xz^2y' = -y^2z^2 + xz^2y'\]
  The proof is now seemingly stuck, as no further progress can be made without explicit information on $y'$. Nonetheless, axiom \lref{ax:dhc} can be applied on the algebraic equality $xy = 1$ to derive \emph{implicit} constraints on $y'$, this gives:
  \[\sequent{}{(xy)' = 0 \rightarrow \boxx{\{x, y : x' = -y \land xy = 1\}}{(xy)' = 0}}\]
  Direct computations give $(xy)' = x'y + xy'$, which is indeed equivalent to $0$ at the initial conditions $x = y = 1, x' = -1, y' = 1$. Thus, axiom \lref{ax:dhc} derives that this is also true along the dynamics, simplifying with $xy = 1$ gives
  \[\sequent{x'y + xy' = 0, xy = 1, x' = -y}{xy' = y^2}\]
  deriving
  \[\sequent{x = y = 1, x' = -1, y' = 1}{\boxx{\{x, y : x' = -y \land xy = 1\}}{xy' = y^2}}\]
  Finally, combining this with our earlier differential computation (axiom \lref{ax:dc}) proves
  \[\sequent{(z^2y)' = -y^2z^2 + xz^2y',  xy' = y^2}{(z^2y)' = -y^2z^2 + y^2z^2 = 0}\]
  Thus, this shows that $z^2y = 1$ is an invariant along the DAE, thereby provably establishing the safety property $y > 0$, as desired. 
\end{example}

\section{DAE Index Reduction}
\label{sec:index_reduction}
Unlike ordinary differential equations, differential-algebraic equations require symbolic preprocessing before deductive verification methods can be applied.
In particular, not just any choice of initial values admits a local solution: the initial data must satisfy the algebraic constraints and the associated dynamics constraints that follow from them.
For instance, consider the system: 
\(
  \dap{x,y}{x^2 + y^2 = 1}.
\)
Assigning arbitrary values to $(x,y)$ with $x^2+y^2=1$ is not sufficient to guarantee the existence of a positive-duration trace, since the differential variables $(x',y')$ must also be initialized consistently, tangent to the unit circle.
This condition can be made explicit by computing the differential of the circle equation: $\D{x^2 + y^2 - 1} \equiv 2xx' + 2yy' = 0$.
Thus, a trace of positive duration satisfying algebraic constraint $x^2 + y^2 = 1$ must also satisfy the induced hidden constraint $2xx' + 2yy' = 0$ all along its evolution.
Revealing these hidden constraints is crucial for DAEs safety proofs.

In the numerical DAE literature, Gear~\cite{Gear1988} showed that differentiation of algebraic constraints exposes these implicit constraints and formalized the concept of the \emph{differentiation index} (or simply index).
The index is the number of times a system must be differentiated to uniquely solve for its differential variables, so a time-stepping procedure can compute a single, well-defined trajectory at each step.
While this standard index concept is adequate for numerical simulation, it silently assumes that the system is well-posed.
Although this assumption is pragmatic, it is inadequate for deductive verification.
In this domain, models are often parametric and underspecified. 
Yet, revealing implicit constraints on differential variables is a prerequisite to enable the full power of differential invariants.
To illustrate how these hidden constraints arise in simple mechanical models, consider the canonical benchmark for index reduction (cf. \cite{Gear1988}):
\begin{example}\label{ex:euclidean_pendulum}
  To illustrate the extraction of hidden constraints, consider a constrained mechanical system, a pendulum in Cartesian coordinates, with variables $\vvX = (x, y, v, w, \lambda)$. 
  The variables $x,y$ are the $x,y$-coordinates, $v,w$ their corresponding velocities and $\lambda$ is a constraint multiplier.
  The constraint of the initial system is given by:
  \[
    \ttF = \begin{pmatrix} 
      x' - v\\ 
      mv' - \lambda x \\
      y' - w \\
      mw' - \lambda y - mg \\
      x^2 + y^2 - l^2
    \end{pmatrix}
  \]
  While the derivatives of the state variables $x,y,v,w$ are explicit, the system is not differentially closed; meaning that differentiation yields new independent algebraic equations that were not present originally.
  Specifically, each trace of positive duration must satisfy a tangency condition to the algebraic constraint $x^2 + y^2 = l^2$.
  We capture this condition and reduce the index by prolonging the system with the differential of the algebraic constraint, $(x^2 + y^2 - l^2)' \equiv 2xx' + 2yy'$. 
  Performing an algebraic reduction by substituting $x' = v$ and $y' = w$, we obtain the velocity constraint $xv + yw = 0$, ensuring the velocity remains tangent to the circle.
  Repeating this process of differentiation and reduction reveals a second hidden constraint that determines the multiplier $\lambda$ algebraically:
  \[
    l^2\lambda + m(v^2 + w^2) - mgy = 0.
  \]
  Intuitively, this equation fixes the multiplier $\lambda$ to the exact value required to prevent the pendulum from flying off the circular path with radius $l$.
  Finally, differentiating this multiplier equation, we obtain an equation for the differential $\lambda'$. 
  By collecting the original dynamics, the discovered hidden constraints and this final equation, we arrive at the system:
 \begin{equation}\label{eq:pendulum_reduced}
   \ttF_{\text{red}} = \begin{pmatrix} 
      x' - v\\ 
      mv' - \lambda x \\
      y' - w \\
      mw' - \lambda y - mg \\
      l^2 \lambda' + 2 m (m - 1)gw + 2 \lambda(vx+mwy)\\
      x^2 + y^2 - l^2 \\
      2xv + 2yw\\
      l^2\lambda + m(v^2 + w^2) -mgy
    \end{pmatrix}
  \end{equation}

  Notice that the system is now differentially closed.
  We demonstrate this by checking the Jacobian of the system's differential subsystem $\ttG$ (defined as the set of equations in \eqref{eq:pendulum_reduced} containing differential variables):
  \[
  \jacobian[\vvXp][\ttG] \equiv \begin{pmatrix}
      1 & 0 & 0 & 0 & 0\\
      0 & 1 & 0 & 0 & 0\\
      0 & 0 & m & 0 & 0\\
      0 & 0 & 0 & m & 0 \\
      0 & 0 & 0 & 0 & l^2 
      \end{pmatrix}
  \]
  Provided that the physical parameters satisfy $m \neq 0$ and $l \neq 0$, the Jacobian is non-singular.
  This ensures that the system now uniquely determines the differential variables, characterizing an implicit ODE.
\end{example}
The informal index reduction procedure sketched above can be formalized within the \dAL proof calculus as a refinement between the original and the reduced system.
\begin{example}\label{ex:euclidean_formal}
  Proving the refinement \(
    \dap{\vvX}{\ttF = \ttZero} \rfeq \dap{\vvX}{\ttF_{\text{red}} = \ttZero},
  \)
  requires the following explicit initial conditions:
  \begin{align*}
    &\Gamma \equiv xx' + yy' = 0 \Land x'v + xv' + y'w + yw' = 0 \\
    &\qquad \Land l^2  \lambda'  + 2mvv' + 2mww' - mgy' = 0.
  \end{align*}
  Geometrically, $\Gamma$ enforces tangency to the algebraic constraints.
  Without it, the refinement fails to hold, because the reduced system would not admit a trace of positive duration.

  We now demonstrate how a single step of differentiation and algebraic reduction is performed within the calculus.
  The \lref{rl:dhc} rule provides the formal justification for refining the augmented differential system:
  \begin{prooftree}
    \ax*[rl:real]{\Gamma}{2xx' + 2yy' = 0}
    \ax*[rl:real]{\ttF = \ttZero}{x^2 + y^2 = l^2}
    \un[rl:dw]{\Gamma}{\fBox[\dap{\vvX}{\ttF = \ttZero}][x^2 + y^2 = l^2]}
    \bi[rl:dhc]{\Gamma}{\dap{\vvX}{\ttF = \ttZero} \rfeq \dap{\vvX}{\ttF = \ttZero \Land 2xx' + 2yy' = 0}}
  \end{prooftree}

  The algebraic reduction is derived with the \lref{ax:dc}~axiom and the \lref{rl:dw}~rule: \begin{prooftree}
    \ax*[rl:real]{}{(\ttF = \ttZero \Land 2xx' + 2yy' = 0) \leftrightarrow (\ttF = \ttZero \Land 2xv + 2yw = 0)}
    \un[ax:dc,rl:dw]{\Gamma}{\dap{\vvX}{\ttF = \ttZero \Land 2xx' + 2yy' = 0} \rfeq \dap{\vvX}{\ttF = \ttZero \Land 2xv + 2yw = 0}}
  \end{prooftree}
  
  By the \lref{ax:tr}~axiom, we may combine the above derivations to obtain:
  \[
    \Gamma \vdash \dap{\vvX}{\ttF = \ttZero} \rfeq \dap{\vvX}{\ttF = \ttZero \Land 2xv + 2y = 0}
  \]
  
  Iterative application of these rules yields the following refinement of the original and the reduced system from \eqref{eq:pendulum_reduced}:
  \[
    \Gamma \vdash \dap{\vvX}{\ttF = \ttZero} \rfeq \dap{\vvX}{\ttF_{\text{red}} = \ttZero}
  \]
  This demonstrates that the \dAL proof calculus is capable of formalizing the index reduction refinement for the Euclidean Pendulum example.
  Furthermore, as the established equivalence between the reduced and the original system is \emph{parametric} in $l$ and $m$, 
  this reduction is provably valid not just for a specific case, but for any instance of these parameters.
\end{example}
The technique of Example~\ref{ex:euclidean_formal} is representative of most DAE index reduction methods.
The process is iterative, operating on a general differential-algebraic system $\dap{\vvX}{\ttF = \ttZero}$ by following these steps at each stage $i$:
\begin{enumerate}
  \item Extract the algebraic part $\ttR^A_i$: This is the set of consequences of $\ttF_i = \ttZero$ that contain no differential variables. \footnote{Formally, $\ttR_i^A$ can be found by computing a basis of the elimination ideal using the lexicographic monomial ordering where the derivatives are ranked higher than state variables (cf. \cite{CoxLittleOShea2025}).}
  \item Augment the system with the differential of $\ttR_i^A$: \[\ttF_{i+1} \equiv \ttF_i \cup \D{\ttR_i^A}.\]
\end{enumerate}
We establish the soundness of this transformation through the following theorem.
It proves that the reduced system is equivalent to the original one, provided that the initial conditions satisfy the hidden constraints.
\ldef{rl:ir}{IR}{Formula}
\begin{theorem}[DAE index reduction]\label{thm:index_reduction}
  Let $\ttR_0^A,\dots,\ttR_m^A$ be differential-free terms.
  Then, the following is a derivable proof rule of \dAL:
  \begin{sequentproof}
      \ax{\Gamma}{\bigwedge_{i=0}^{m} \D{\ttR_i^A} = \ttZero}
      \ax{}{\bigwedge_{i=0}^m(\ttF_i = \ttZero \rightarrow \ttR_i^A = \ttZero)}
      \bi[rl:ir]{\Gamma}{\dap{\vvX}{\ttF_{m} = \ttZero} \rfeq \dap{\vvX}{\ttF_0 = \ttZero}}
  \end{sequentproof}
\end{theorem}
\begin{proof}
  See Appendix \ref{sec:index_reduction_proof}.
\end{proof}
This theorem establishes that our calculus is relatively complete with respect to the symbolic index reduction procedure; meaning every valid step of the reduction corresponds to a derivable sequence of proof rules in \dAL.

\section{Related work}
\paragraph{Deductive Verification of Differential Equations} The deductive verification of differential equations aims to axiomatize all necessary principles using a small set of core axioms, and deriving other desired properties (e.g. safety, liveness) with accompanying trustworthy syntactic proofs using the sound axioms. Such approaches include differential dynamic logic (\dL) and its variants \cite{DBLP:journals/jar/Platzer08,DBLP:conf/cade/HellwigP25,DBLP:conf/lics/LoosP16,DBLP:journals/logcom/Platzer10}, some of which can be proof-checked in KeYmaera X \cite{Fulton_Mitsch_Quesel_Volp_Platzer_2015}. Sophisticated proof-rules and partial completeness results have also been established in the case for ordinary differential equations \cite{DBLP:journals/jacm/PlatzerT20,compact_IVP}. Axiomatic approaches to DAEs are relatively less explored. Earlier works \cite{DBLP:journals/logcom/Platzer10,DBLP:conf/cade/HellwigP25} considered a direct extension of the reachability-based semantics in \dL to DAEs, and therefore do not provide a refinement-calculus capable of directly expressing index reductions. Earlier work \cite{DBLP:conf/cade/HellwigP25} also assumes traces of DAEs to be real-analytic, which is a much stronger condition than the $C^1$ regularity imposed in this article.
\paragraph{Numerical and structural Index Reduction} Classical index reduction techniques such as those by Gear~\cite{Gear1988}, Pantelides~\cite{Pantelides1988}, and Pryce~\cite{Pryce2001} address the consistent initialization of DAEs for numerical simulation.
In contrast, our goal is to reason about parametric and underspecified models within our axiomatic framework via symbolic invariant reasoning.
By revealing the hidden constraints of a DAE, we can reason about safety entirely without numerical integration.
\paragraph{Differential algebra}
Differential algebra offers a formal treatment of differential equations through the study of differential ideals \cite{ritt1950differential,Kolchin1973}.
Later, these ideas were operationalized into effective algorithms \cite{DBLP:conf/issac/BoulierLOP95,DBLP:conf/snsc/Hubert01a}, providing a rigorous foundation for symbolic index reduction.
However, these approaches are purely structural and do not account for the semantics of continuous traces.
Consequently, they lack the logical framework to reason about reachability or safety properties.
We address this by integrating index reduction into a deductive proof calculus.
\paragraph{Formal semantics of Modeling Languages}
The work of Benveniste et al. \cite{DBLP:conf/hybrid/BenvenisteCEGOP17,DBLP:journals/arc/BenvenisteCM20} most closely aligns with our framework in terms of modeling power, as they formally address the interplay between DAEs and discrete mode transitions.
However, their primary concern is the correctness of compilation and execution for DAE-based simulation software.
By contrast, we go beyond execution and provide an axiomatic framework for deductive verification.

\section{Conclusion}
This paper presented an axiomatic framework $\dAL$ for the deductive verification of general DAEs. By leveraging a novel trace-based semantics, $\dAL$ supports sound comparison of DAEs through refinements, enabling the incremental verification of complicated DAEs via iterative simplifications. In particular, index reductions of DAEs can be completely internalized in $\dAL$, with supporting syntactic proofs of correctness at each step of the reduction. For future work, it would be interesting to implement a proof-checker for $\dAL$.

\begin{credits}
\subsubsection{\ackname}
This work has been supported by an Alexander von Humboldt Professorship and the Deutsche Forschungsgemeinschaft (DFG, German Research Foundation) - SFB 1608 - 501798263.

\subsubsection{\discintname}
The authors have no competing interests to declare that are
relevant to the content of this article.
\end{credits}

\printbibliography
\appendix
\section{Deferred Proofs}

\subsection{Conservative Extension.}
\label{sec: appendix_conservative_extension}
This section completes the proof of Proposition \ref{prop: conservative_extension}, that $\dAL$ is a conservative extension of the reachability-based semantics of $\dL$ \cite{DBLP:journals/jar/Platzer08}.
\begin{proof}
  Let \((\sO, \sN) \in \States \times \States\) be a pair of arbitrary states.
  We establish the equivalence with respect to reachability by structural induction on programs.
  Let \(\pA,\pB\) be two programs satisfying the induction hypothesis.
  \begin{description}
    \item[Case \(\assign\)] \hfill \\
    \((\Rightarrow)\) Let \((\sO,\sN) \in \semdL{\assign}\).
    By the relational semantics of \dL, we have \(\sN = \stateUpd[\sO][\vX][\semS{\tE}{\sO}]\), i.e. $\sN$ is the state $\sO$ with variable $x$ having value $\semS{\tE}{\sO}$. From the semantics of discrete assignment, the zero-duration trace satisfies \(\flPoint[\sN] \in \semPS[\assign]\).
    By construction, it holds \(\delta_{\sN}(0) = \sN\), completing the forward direction.

    \smallskip
    \((\Leftarrow)\) Conversely, let \(\flA \in \semPS[\assign]\) with \(\flAT = \sN\).
    By the trace-based semantics, we have \(\flA = \delta_\sN\) and \(\sN = \stateUpd[\sO][\vX][\semS{\tE}{\sO}]\).
    Thus, by the \dL-semantics, \((\sO,\sN) \in \semdL{\assign}\), completing this case.
    \item[Case \(\test\)] \hfill \\
    \((\Rightarrow)\) Let \((\sO,\sN) \in \semdL{\test}\).
    By the \dL-semantics, we have \(\sO = \sN\) and \(\sO \in \sem{\fQ}\).
    By the trace-based semantics, \(\sO \in \sem{\fQ}\) implies \(\delta_\sO \in \semPS[\test]\).
    Since \(\delta_\sO(0) = \sO = \sN\), there exists a trace reaching the terminal state \(\sN\), as required.
    \smallskip

    \((\Leftarrow)\) Conversely, let \(\flA \in \semPS[\test]\) with \(\flAT = \sN\).
    By the trace-based semantics of test, \(\flA\) is equal to the unique zero-duration trace \(\delta_\sO \in \semPS[\test]\).
    Thus, \(\sO = \sN\) and \(\sO \in \sem{\fQ}\), yielding \((\sO,\sN) \in \semdL{\test}\), as required.

    \item[Case \(\pA;\pB\)] \hfill \\
    \((\Rightarrow)\) Assume \((\sO,\sN) \in \semdL{\pA ; \pB}\).
    By the \dL-semantics, there exists an intermediate state \(\sM \in \States\) such that \((\sO,\sM) \in \semdL{\pA}\) and \((\sM,\sN) \in \semdL{\pB}\).
    Applying our inductive hypothesis to these pairs, there exist traces \(\flA \in \semPS\) and \(\flB \in \semPS[\pB][\sM]\) with \(\flAT = \sM\) and \(\flBT = \sN\) respectively. 
    Let \(\flC = \flConcat{\flA}{\flB}\) be the concatenation of these traces.
    By the semantics of sequential composition, we have \(\flC \in \semPS[\pA ; \pB]\).
    Since \(\flCT = \flBT = \sN\), the forward direction holds.
    \smallskip

    \((\Leftarrow)\) Conversely, let \(\flC \in \semPS[\pA; \pB]\) with \(\flCT = \sN\). 
    By the semantics of sequential composition, there exist \(\flA \in \semPS[\pA]\) and \(\flB \in \semPS[\pB][\flAT]\) such that \(\flC = \flA \cdot \flB\).
    Applying the inductive hypothesis to these traces, we conclude \((\sO,\flAT) \in \semProg_\dL\) and \((\flAT, \sN) \in \semdL{\pB}\).
    By the \dL-semantics, the state \(\flAT\) serves as a witness for \((\sO,\sN) \in \semdL{\pA;\pB}\), completing the case.

    \item[Case \(\pA \cup \pB\)] \hfill \\
    \((\Rightarrow)\) Let \((\sO,\sN) \in \semdL{\pA \cup \pB}\).
    By the relational semantics of \dL, we have either \((\sO,\sN) \in \semdL{\pA}\) or \((\sO,\sN) \in \semdL{\pB}\).
    From the IH, there exists a trace reaching \(\sN\) in either \(\semPS\) or \(\semPS[\pB]\).
    Thus, there exists a trace in \(\semPS \cup \semPS[\pB] = \semPS[\pA \cup \pB]\), as required.

    \smallskip
    \((\Leftarrow)\) Let \(\flA \in \semPS[\pA \cup \pB]\) with terminal state \(\flAT = \sN\).
    By the semantics of non-deterministic choice, \(\flA\) is contained in either \(\semPS\) or \(\semPS[\pB]\).
    Applying the IH, yields \((\sO,\sN) \in \semdL{\pA \cup \pB}\).

    \item[Case \(\rep\)] \hfill \\
    Recall that $\rep$ represents (finite) unbounded composition, thus this case is proved by another induction on the equivalence for $\pA^{n}$ for all $n \in \N$, the base case of $n = 0$ is trivial. \\ 
    \((\Rightarrow)\) Let $(\sO,\sN) \in \semdL{\alpha^{n + 1}}$, by definition of sequential composition this implies the existence of an intermediate state $\sM$ such that $(\sO, \sM) \in \semdL{\alpha}$ and $(\sM, \sN) \in \semdL{\alpha^n}$. The claim now follows by first applying the inductive hypothesis of equivalence at $n$ compositions, followed by the inductive hypothesis for equivalence under sequential compositions.

    \smallskip
    \((\Leftarrow)\)
    Similarly, let $\flC \in \semPS[\pA^{n + 1}]$ be an arbitrary trace with terminal state $\flCT = \sN$, by the $\dAL$ semantics for composition, we can decompose $\flC$ as \(\flC = \flConcat{\flA}{\flB}\) with $\flA \in \semPS{\alpha}, \flB \in \semPS[\alpha^n]$. Furthermore, these traces must satisfy $\flA(0) = \flC(0), \flBT = \flCT$, from which the desired claim follows by applying the inductive hypothesis for equivalence of composition.
  \end{description}
\end{proof}

\subsection{Soundness.}\label{sec:appendix_soundness_proof}
We use the following well-known fundamental results from analysis and linear algebra:
\begin{theorem}[Mean Value Theorem]\label{thm:mean-value}
Let $f \colon [a,b] \to \mathbb{R}$ be a function that is continuous on the closed interval $[a,b]$ and differentiable on the open interval $(a,b)$. Then there exists a point $\xi \in (a,b)$ such that
\[
f'(\xi) = \frac{f(b) - f(a)}{b - a}.
\]
\end{theorem}
\begin{theorem}\label{thm:det_unique_solution}
Let $A \in \mathbb{R}^{n \times n}$ be a square matrix. 
If $\det(A) \neq 0$,
then for every $b \in \mathbb{R}^n$ the linear system
\[
A x = b
\]
has a unique solution $x \in \mathbb{R}^n$.
\end{theorem}
\begin{lemma}[Continuity of the matrix inverse]
\label{lem:matrix_inverse_continuity}
Let $J \subseteq \mathbb{R}$ be an interval, and let
$\mA \colon J \to \mathbb{R}^{n \times n}$ be a continuous mapping such that
$\det(\mA(t)) \neq 0$ for all $t \in J$.
Then the mapping
\[
t \mapsto \mA(t)^{-1}
\]
from $J$ to $\mathbb{R}^{n \times n}$ is continuous.
\end{lemma}
\begin{theorem}[\cite{Walter_1998}]
\label{thm:linear_ode_existence}
Let $J \subseteq \mathbb{R}$ be an interval, let
$\mA \colon J \to \mathbb{R}^{n \times n}$ and
$\vvB \colon J \to \mathbb{R}^n$ be continuous functions, and let $t_0 \in J$
and $\vvX_0 \in \mathbb{R}^n$.
Then the initial value problem
\[
\vvXp(t) = \mA(t)\vvX(t) + \vvB(t), \qquad \vvX(t_0) = \vvX_0
\]
has a unique solution $\vvX \colon I \to \mathbb{R}^n$.
\end{theorem}
\begin{theorem}[Implicit Function Theorem]
Let \( U \subseteq \mathbb{R}^n \times \mathbb{R}^m \) be open, and let \( F \in C^k(U; \mathbb{R}^m) \) for some \( k \geq 1 \). Suppose \( (u_0, v_0) \in U \) satisfies \( F(u_0, v_0) = 0 \), and the Jacobian \( J_y F(u_0, v_0) \in \mathbb{R}^{m \times m} \) with respect to \(y\) is invertible.

Then there exist open neighborhoods \( V \subseteq \mathbb{R}^n \) of \( u_0 \) and \( W \subseteq \mathbb{R}^m \) of \( v_0 \), and a unique \( C^k \) function \( g \colon V \to W \) such that
\[
F(u, g(u)) = 0 \quad \text{for all } u \in V,
\]
and for all \( (u,y) \in V \times W \),
\[
F(u,y) = 0 \iff y = g(u).
\]
Moreover, the derivative of \( g \) is given by
\[
Dg(u) = -\big(J_y F(u, g(u))\big)^{-1} J_x F(u, g(u)).
\]
\end{theorem}

The following local existence lemma is well-known, its proof is provided for completeness.

\begin{lemma}[Local solution of an implicit ODE]\label{lem:implicit_ode}
Let $\ttF:\mathbb{R}^n\times\mathbb{R}^n\to\mathbb{R}^n$ be of class $C^1$.
Assume there exist $\ttU_0,\ttV_0\in\mathbb{R}^n$ such that \(
\ttF(\ttU_0,\ttV_0)=0 \)
and the Jacobian matrix \(J_{\ttV}\ttF(\ttU_0,\ttV_0)\) is non-singular.

Then there exist open sets $U,V\subseteq\mathbb{R}^n$ containing $\ttU_0$ and $\ttV_0$ respectively,
and a constant \(\varepsilon > 0\) such that the initial value problem
\[
\ttU(0)=\ttU_0,\qquad \ttU'(0)=\ttV_0,\qquad \ttF\bigl(\ttU(t),\ttU'(t)\bigr)=0 \ \ (t\in J)
\]
admits a unique solution $\ttU\in C^2((-\varepsilon,\varepsilon),\mathbb{R}^n)$ satisfying $\ttU(t)\in U$ and $\ttU'(t)\in V$ for all \(t \in (-\varepsilon,\varepsilon)\).
\end{lemma}

\begin{proof}
  By construction, $\ttF(\ttU_0,\ttV_0)=0$.
    Moreover, by assumption the Jacobian of $\ttF$ with respect to its second argument is non-singular at $(\ttU_0,\ttV_0)$, i.e.,
    \[
      \det\bigl(J_{\ttV}\ttF(\ttU_0,\ttV_0)\bigr) \neq 0 .
    \]

    From the Implicit Function Theorem, there exist open sets $U,V \subseteq \mathbb{R}^n$
    with $\ttU_0 \in U$ and $\ttV_0 \in V$, and a unique function $\ttG \in C^1(U,V)$
    such that
    \[
      \ttF(\ttU,\ttG(\ttU)) = 0 \qquad \text{for all } \ttU \in U,
    \]
    and $\ttG(\ttU_0)=\ttV_0$.

    Consider the ordinary differential equation
    \begin{equation}\label{eq:ag_ode}
      \ttU'(t) = \ttG(\ttU(t)), \qquad \ttU(0)=\ttU_0 .
    \end{equation}
    Since $\ttG \in C^1(U,V)$, it is locally Lipschitz on $U$.
    Hence, by the Picard--Lindel\"of theorem \cite{Walter_1998} there exists an open interval
    $J$ containing $0$ and a unique solution $\ttU \in C^1(J,U)$ of~\eqref{eq:ag_ode}.

    Because both $\ttG$ and $\ttU$ are $C^1$, the chain rule implies that the composition
    $t \mapsto \ttG(\ttU(t))$ is of class $C^1$ on $J$.
    Using~\eqref{eq:ag_ode}, we obtain $\ttU'(\cdot) = \ttG(\ttU(\cdot)) \in C^1(J,\mathbb{R}^n)$.
\end{proof}

The proof of soundness for $\dAL$ also uses the following lemmas. 

\begin{lemma}[Regular Differential]
  \label{lem:diff}
  Let \(e\) be a term with $\vvXp \notin \tE$. For any \(\vvX\)-regular trace \(\flA \in \flRegular\)
  of duration \(T_\flA > 0\), the mapping \(t \mapsto \semS{\tE}{\flAt}\) is continuously differentiable on \(\flInt\).
  Furthermore, for all \(t \in \flInt\), it holds that: 
  \[
    \ddt \semS{\tE}{\flAt} = \semS{\D{\tE}}{\flAt}.
  \]
\end{lemma}

\begin{proof}
  First note that $C^1$ regularity follows from the equivalence $\ddt \semS{\tE}{\flAt} = \semS{\D{\tE}}{\flAt}$, as the RHS is a term and therefore continuous as $\flA$ is continuous. This equivalence can be established by inducting on the structure of $\tE$ and the axiomatization of differential terms in Section \ref{sec:proof_calculus}, a detailed proof can be found in earlier work axiomatizing differentials in \dL \cite{DBLP:journals/jar/Platzer17}.
\end{proof}

\begin{lemma}
  \label{lem:det_correct}
  Let \(\mA \in \Terms^{n \times n}\) be a square matrix of terms.
  Then, for all states \(\sO \in \States\),
  \(\det{\semS{\mA}{\sO}} =  \semS{\mdet{\mA}}{\sO}.\)
\end{lemma}

\begin{proof}
  This follows directly from the fact that $\mdet{\mA}$ was defined (Definition \ref{def:det}) via its cofactor expansion. 
\end{proof}

\begin{lemma}\label{lem:c_flow}
  Let \(\ttE\) be a vectorial term.
  For any state \(\sO \in \sem{\ttE = \ttZero \Land \mdet{\jacobian[\vvXp][\ttE]}\neq 0}\), there exists an interval \(J = (-\epsilon,\epsilon)\) for some \(\varepsilon > 0\) and a unique \(C^2\) trace \(\flC : J \rightarrow \States\) such that
  \begin{enumerate}
    \item \(\flC(0) = \sO\),
    \item \(\flC(J) \subseteq \sem{\ttE = \ttZero}\),
    \item \(\flC(t)(v) = \sO(v)\) for \(v \not\in \vvZ \cup \vvZp\).
  \end{enumerate}
\end{lemma}
\begin{proof}
   Let $\ttU_0 = \sO(\vvX)$ and $\ttV_0 = \sO(\vvX')$ and
    define the induced semantic mapping $\ttF : \mathbb{R}^n \times \mathbb{R}^n \to \mathbb{R}^n$ by
    \[
      \ttF(\ttU,\ttV)
      \equiv \semS{\ttE}{\sO_{\ttU,\ttV}},
    \]
    where \(\sO_{\ttU,\ttV}\) is the modified sate \(\sO\) such that \(\sO(\vvX) = \ttU\) and \(\sO(\vvXp) = \ttV\).
    By Lemma \ref{lem:det_correct}, we conclude \(\det \mJ_\ttV \ttF(\ttU_0,\ttV_0) = \det \semS{\jacobian[\vvXp][\ttE]}{\sO} = \semS{\mdet \jacobian[\vvXp][\ttE]}{\sO} \neq 0\).
    Applying Lemma \ref{lem:implicit_ode},
    we obtain \(J = (-\varepsilon,\varepsilon)\) for some \(\varepsilon > 0\) on which the initial value problem
    \[
    \ttU(0)=\ttU_0,\quad \ttU'(0)=\ttV_0,\quad \ttF\bigl(\ttU(t),\ttU'(t)\bigr)=0
    \]
    admits a unique solution \(\ttU \in C^2(J,\R^n)\).
\end{proof}

Finally, soundness of $\dAL$ (Theorem \ref{thm: sound axiomatization}) is established by proving that each axiom is semantically valid:
\begin{proof}[Soundness of Refinement Axioms]
  \begin{itemize}
    \item[\lref{ax:r}] Let \(\sO \in \sem{\pA \rfle \pB}\) and \(\sO \in \sem{\fBox[\pB]}\).
    To show \(\sO \in \sem{\fBox}\), fix an arbitrary trace \(\flA \in \semPS\).
    By the refinement assumption, there exists a \(\flB \in \semPS[\pB]\) such that \(\coincide\), which implies \(\flTA = \flTB\).
    Since \(\sO \in \sem{\fBox[\pB]}\), it follows that \(\flB(\flTA) \in \sem{\fP}\).
    By the Coincidence Lemma, it suffices to show \(\coincide[\FV{\fP}][\flA][\flB]\).

    Given the syntactic side condition \(\FV{\fP} \subseteq \vvX \cup \left(\BV{\pA} \cup \BV{\pB}\right)^\complement\), we consider an arbitrary variable \(v \in \FV{\fP}\) and split into two cases:
    \begin{enumerate}
      \item If \(v \in \vvX\): Then, we have \(\flAt(v) = \flB(t)(v)\) for all \(t \in \flInt\), by the refinement assumption.
      \item If \(v \in \left(\BV{\pA} \cup \BV{\pB}\right)^\complement\):
      By the Bound Effect Lemma, both traces preserve the initial state, i.e., \(\flAt(v) = \sO(v)\) and \(\flB(t)(v) = \sO(v)\) for all \(t \in \flInt\).
      Consequently, \(\flAt(v) = \flB(t)(v)\) for all \(t \in \flInt\).
    \end{enumerate}
    In both cases, \(\flA\) and \(\flB\) coincidence. Thus, by a pointwise application of Lemma~\ref{lem:coincidence} we conclude \(\flAT \in \sem{\fP}\), completing the proof.
    \item[\lref{ax:tr}] Let \(\sO \in \sem{\pA \rfle \pB}\) and \(\sO \in \sem{\pB \rfle \pC}\). 
    To show \(\sO \in \sem{\pA \rfle \pC}\), fix an arbitrary trace \(\flA \in \semPS\).
    Our goal is to demonstrate that there exists a trace \(\flC \in \semPS[\pC]\), such that \(\coincide[\vvX][\flA][\flC]\).

    By the first refinement assumption, there exists a trace \(\flB\in \semPS[\pB]\) such that \(\coincide\).
    Invoking the second assumption for this \(\flB\), there exists a trace \(\flC \in \semPS[\pC]\) with \(\coincide[\vvX][\flB][\flC]\).
    Transitivity yields \(\coincide[\vvX][\flA][\flC]\), which establishes the desired result.
  \end{itemize}
\end{proof}

\begin{proof}[Soundness of Differential Axioms]
  We prove the validity of each axiom individually:
  \begin{itemize}
    \item[\lref{ax:dw}] Let \(\sO \in \States\). To show \(\sO \in \sem{\fBox[\dap][\fF]}\), fix an arbitrary trace \(\flA \in \semPS[\dap]\).
    
    By the semantics of DAPs, we have \(\flA(\flInt) \subseteq \sem{\fF}\), which immediately implies \(\flAT \in \sem{\fF}\), as required by the box modality.
    \item[\lref{ax:dileq}]
    Let \(\sO \in \sem{\tE \preceq 0}\).
    To prove \(\sO \in \sem{\fBox[\dap{\vvX}{\D{\tE} \leq 0}][\tE \preceq 0]}\), we fix an arbitrary \(\vvX\)-regular trace \(\flA \in \semPS[\dap{\vvX}{\D{\tE} \leq 0}]\) with duration \(\flTA\).
    The argument proceeds by case analysis on the duration \(\flTA\):
    \begin{enumerate}
      \item If \(\flTA = 0\): The result is immediate as \(\flA(0) = \sO\) and \(\sO \in \sem{\tE \preceq 0}\) by assumption.
      \item If \(\flTA > 0\):
        Since \(\flA\) is \(\vvX\)-regular and $\tE$ is differential-free, the valuation of the term \(\tE\) along the trace is continuous on \(\flInt\) and differentiable on \((0,\flTA)\) (Lemma~\ref{lem:diff}).
      By the Mean Value Theorem (Theorem~\ref{thm:mean-value}), there exists a \(\xi \in (0,\flTA)\) such that:
      \[
        \semS{\tE}{\flAT} = \flTA \, \ddt \semS{\tE}{\flA(\xi)} + \semS{\tE}{\flA(0)}.
      \]
      Again, since $\tE$ is differential free, by applying the Differential Lemma~\ref{lem:diff}, we obtain:
      \[
        \semS{\tE}{\flAT} = \flTA \, \semS{\D{\tE}}{\flA(\xi)} + \semS{\tE}{\flA(0)}. \label{eq:mvt}
      \]
      Since \(\semS{\tE}{\flA(0)} \preceq 0\) by assumption and \(\semS{\D{\tE}}{\flA(\xi)} \leq 0\) by the differential constraint, it follows that \(\semS{\tE}{\flAT} \preceq 0\).
    \end{enumerate}
    In both cases, the terminal state reached by the trace \(\flA\) satisfies \(\tE \preceq 0\).
    As \(\flA\) was chosen arbitrarily, it follows by the semantics of the box modality that \(\sO \in \sem{\fBox[\dap{\vvX}{\D{\tE} \leq 0}][\tE \preceq 0]}\).
    \item[\lref{ax:dhc}] Let \(\sO \in \sem{\D{\tE} = 0}\).
    To show \(\sO \in \sem{\fBox[\dap{\vvX}{\tE = 0}][\D{\tE} = 0]}\), fix an arbitrary \(\vvX\)-regular trace \(\flA \in \semPS[\dap{\vvX}{\tE = 0}]\) with duration \(\flTA\).
    We distinguish two cases based on the duration \(\flTA\):
    \begin{enumerate}
      \item If \(\flTA = 0\): The goal is immediate as \(\flA(0) = \sO\) and \(\sO \in \sem{\D{\tE} = 0}\) by assumption.
      \item If \(\flTA > 0\): Since \(\flA\) is \(\vvX\)-regular, the valuation of the term \(\tE\) along the trace is differentiable on \((0,\flTA)\) (Lemma~\ref{lem:diff}).
      By the semantics of DAPs, we have \(\flA(\flInt) \subseteq \sem{\tE = 0}\).
      It follows that for every \(t \in \flInt\), the valuation satisfies \(\semS{\tE}{\flAt} = 0\).
      Differentiating both sides of this equation yields \(\ddt \semS{\tE}{\flAt} = 0\).
      By the Differential Lemma~\ref{lem:diff}, this implies \(\semS{\D{\tE}}{\flAt} = 0\) for \(t \in \flInt\).
    \end{enumerate}
    In both cases, the terminal state satisfies \(\D{\tE} = 0\).
    As the choice of trace was arbitrary, we conclude \(\sO \in \sem{\fBox[\dap{\vvX}{\tE = 0}][\D{\tE} = 0]}\).
  \end{itemize}
\end{proof}

\begin{proof}[Soundness of Differential Refinement Axioms]
  We prove the validity of each axiom separately:
  \begin{itemize}
    \item[\lref{ax:dc}] Let \(\sO \in \sem{\fBox[\dap][\fR]}\).
    To show that \(\sO \in \sem{\dap \rfle \dap{\vvX}{\fR}}\), fix an arbitrary \(\vvX\)-regular trace \(\flA \in \semPS[\dap]\).
    By the reflexivity of the coincidence relation (\(\coincide[\vvX][\flA][\flA]\)), 
    it suffices to show that \(\flA\) is itself satisfies \(\flA \in \semPS[\dap{\vvX}{\fR}]\).

    By the semantics of DAPs, we already have \(\flA(\flInt) \subseteq \sem{\fF}\).
    To show \(\flA(\flInt) \subseteq \sem{\fR}\), consider an arbitrary \(\tau \in \flInt\).
    Since DAPs allow for arbitrary durations, the restriction of \(\flA\) to \([0,\tau]\) is a valid execution of \(\dap\) starting at \(\sO\).
    By the box modality assumption, it follows that \(\flA(\tau) \in \sem{\fR}\).
    As \(\tau\) was arbitrary, we conclude \(\flA(\flInt) \subseteq \sem{\fR}\).
    Thus, \(\flA \in \semPS[\dap{\vvX}{\fR}]\), completing the proof.
    \item[\lref{ax:dr}] Let \(\sO \in \States\).
    To show \(\sO \in \sem{\forall \vvZ \forall \vvZp \left(
    \dap{\vvX,\vvZ}{\fR \Land \fF} \rfle \dap\right)}\),
    fix arbitrary vectors \(\ttU, \ttV \in \R^m\) and define \(\sN = \stateUpd[\sO][(\vvZ,\vvZp)][(\ttU,\ttV)]\).
    Now, it suffices to show that for an arbitrary trace \(\flA \in \semPS[\dap{\vvX,\vvZ}{\fR \Land \fF}][\sN]\), there exists a trace \(\flB \in \semPS[\dap][\sN]\) such that \(\coincide\).

    Define the candidate trace \(\flB \in \flSpace\) with duration \(\flTB = T_\flA\) and for \(t \in [0,\flTB]\):
    \begin{align*}
      \flB(t)(v) \equiv \begin{cases}
          \flAt(v) & \text{ if } v \in \vvX\cup\vvXp, \\
          \sN(v) & \text{ otherwise. }
      \end{cases}
    \end{align*}
    To verify \(\flB \in \semPS[\dap][\sN]\), we need to check two properties:
    \begin{enumerate}
      \item \(\vvX\)-regularity: The regularity of \(\flB\) is immediate by the \((\vvX,\vvZ)\)-regularity of the original trace \(\flA\).
      \item \(\flA([0,\flTB]) \subseteq \sem{\fF}\): Recall the syntactic side condition: \(\vvZ,\vvZp \notin \fF\).
      By construction, we have \(\coincide[\{\vvZ,\vvZp\}^\complement]\).
      Since the original trace satisfies the augmented constraint \(\flA(\flInt) \subseteq \sem{\fR \Land \fF} \subseteq \sem{\fF}\), applying the Coincidence Lemma~\ref{lem:coincidence} pointwise ensures that \(\flB([0,\flTB]) \subseteq \sem{\fF}\).
    \end{enumerate}
    
    Combining these properties, we have \(\flB \in \semPS[\dap][\sN]\).
    As the choice of \(\flA\) was arbitrary, this completes the proof. \qed
    \item[\lref{ax:dg}] Let \(\sO \in \sem{\fBox[\dap][\det{\mA} \neq 0]}\).
    To establish the refinement, fix an arbitrary vector \(\ttU \in \R^n\).
    We first demonstrate the existence of a unique derivative vector \(\ttV\) at the initial state \(\sO\).
    Consider the following linear system 
    \begin{equation}\label{eq:sys}
      \semS{\mA}{\sO} \, \ttV = \semS{\mB}{\sO}\ttU + \semS{\ttC}{\sO}.
    \end{equation}
    Assume \(\sO \in \sem{\fF}\).
    Otherwise, there is no trace \(\flA \in \semPS[\dap]\) and the refinement holds vacuously.
    By the semantics of the box modality and \(\sO \in \sem{\fF}\), the zero-duration trace satisfies \(\delta_\sO \in \semPS[\dap]\).
    Thus, the initial state satisfies \(\sO \in \sem{\mdet{\mA} \neq 0}\).
    Hence, by Lemma \ref{lem:det_correct}, \(\det{\semS{\mA}{\sO}} = \semS{\mdet{\mA}}{\sO} \neq 0\).
    By Theorem~\ref{thm:det_unique_solution}, the system \eqref{eq:sys} possesses a unique solution \(\ttV \in \R^n\).

    Define the candidate state \(\sN \equiv \stateUpd[\sO][(\vvZ,\vvZp)][(\ttU,\ttV)]\).
    It suffices to show that for every \(\vvX\)-regular trace \(\flA \in \semPS[\dap][\sN]\), there exists a (\(\vvX,\vvZ\))-regular trace \(\flB \in \semPS[\dap{\vvX,\vvZ}{
        \fF \Land 
          \mA \vvZp = \mB\vvZ + \ttC}][\sN]\) 
    such that \(\coincide\).

    By the box modality assumption and Lemma \ref{lem:det_correct}, we have \(\semS{\mdet{\mA}}{\flAt} = \det{\semS{\mA}{\flAt}} \neq 0\) for all \(t \in \flInt\). 
    Consequently, the matrix \(\semS{\mA}{\flAt}\) is invertible on the interval \(\flInt\). 
    Further, by the Differential Lemma~\ref{lem:diff}, the mapping $t \mapsto \semS{\mA}{\flAt}$ is continuous.
    Thus, by Lemma~\ref{lem:matrix_inverse_continuity} the mapping \(t \mapsto (\semS{\mA}{\flAt})^{-1}\) is a continuous function.

    We define the time-dependent linear differential equation to determine the trajectory of \(\vvZ\):
    \begin{equation}\label{eq:ode}
      \vvYp(t) = \mE(t)\vvY(t) + \mathbf{f}(t), \quad \vvY(0) = \ttU,
    \end{equation}
    where \(\mE(t) = (\semS{\mA}{\flAt})^{-1}\semS{\mB}{\flAt}\) and \(\ttF(t) = (\semS{\mA}{\flAt})^{-1} \semS{\ttC}{\flAt}\).
    Since the mappings \(t \mapsto \mE(t)\) and \(t \mapsto \ttF(t)\) are continuous on the compact interval \(\flInt\), by Theorem~\ref{thm:linear_ode_existence}, equation~\eqref{eq:ode} has a unique solution \(\vvY : \flInt \rightarrow \R^n\) of class $C^1$.

    We define the refinement trace \(\flB \in \flSpace\) with duration \(\flTB = T_\flA\) by
    \begin{align*}
      \flBt(\vvZ) &\equiv \vvY(t), \\
      \flBt(\vvZp) &\equiv \vvYp(t), \\
      \flBt(v) &\equiv \flAt(v) \text{ for } v \not\in \vvZ \cup \vvZp.
    \end{align*}
    We show that $\flB$ satisfies the required properties:
    \begin{enumerate}
      \item Differential constraint: By construction, equation \eqref{eq:ode} implies for all $t \in \flInt[\flB]$:
        \[
          \semS{\vvZp}{\flBt} = (\semS{\mA}{\flAt})^{-1}\semS{\mB}{\flAt} \semS{\vvZ}{\flBt} + (\semS{\mA}{\flAt})^{-1} \semS{\ttC}{\flAt}.
        \]
      Also, by $\vvZ,\vvZp \notin \mA, \vvB, \vvC$, applying coincidence of terms yields
      \[
          \semS{\vvZp}{\flBt} = (\semS{\mA}{\flBt})^{-1}\semS{\mB}{\flBt} \semS{\vvZ}{\flBt} + (\semS{\mA}{\flBt})^{-1} \semS{\ttC}{\flBt}.
        \]
    Finally, rearranging this equation and folding the semantics of terms, we obtain \(\flB([0,\flTB]) \subseteq \sem{\mA \vvZp = \mB \vvZ + \ttC}\).

    Also, by $\coincide[(\vvZ \cup \vvZp)^\complement]$ and $\vvZ,\vvZp \notin \fF$, applying the Coincidence Lemma~\ref{lem:coincidence} pointwise yields \(\flB([0,\flTB]) \subseteq \sem{\fF}\).
  \item $(\vvX,\vvZ)$-regularity: On $\vvX$ the regularity of $\flB$ is immediate by construction.
    The regularity of $t \mapsto \flBt(\vvZ)$ follows from the regularity of $t \mapsto \vvY(t)$.
    Further, by construction we have 
    \[
      \ddt \flBt(\vvZ) = \vvYp(t) = \flBt(\vvZp).
    \]
    \end{enumerate}
    
    Finally, combining these properties we have \(\flB \in \semPS[\dap{\vvX,\vvZ}{\fF \Land \mA \vvZp = \mB\vvZ + \ttC}][\sN]\),
    finishing the proof. 
    \item[\lref{ax:ag}] Let \(\sO \in \sem{\fBox[\dap{\vvX}{\fF}][(\ttE = \ttZero \Land \mdet{\jacobian[\vvXp][\ttE]}\neq 0)] }\).
    To establish the claim, it suffices to show that there exist vectors \(\ttR_1,\ttR_2 \in \R^m\) such that the modified state \(\sN = \stateUpd[\sO][(\vvX,\vvXp)][(\ttR_1,\ttR_2)]\) satisfies the refinement relation \[\dap
    \rfle
    \dap{\vvX, \vvZ}{\fF \Land \vvZ = \ttG}.\]

    Assume \(\sO \in \sem{\fF}\).
    Otherwise, there is no \(\vvX\)-regular trace with \(\flA \in \semPS[\dap]\) and the claim holds vacuously.

    Thus, by \(\sO \in \sem{\fF}\) it holds \(\delta_\sO \in \semPS[\dap]\).
    Crucially, the box assumption implies \(\sO = \delta_\sO(0) \in \sem{\ttE = \ttZero \Land \mdet{\jacobian[\vvXp][\ttE]}\neq 0}\).
    Applying Lemma \ref{lem:c_flow}, there exists a unique \(C^2\) trace \(\flC : J \rightarrow \States\) with \(\flC(0) = \sO\).
   
    Hence, by the \(C^2\) regularity of \(\flC\), the valuation \(t \mapsto \semS{\ttG}{\flC(t)}\) is continuously differentiable on \(J\).
    Therefore, \(t \mapsto \ddt \semS{\ttG}{\flCt}\) is well-defined on \(J\).

    We now define a candidate state $\sN \in \States$ by
    \begin{align*}
      \sN(\vvZ) &\equiv \semS{\ttG}{\sO},\\
      \sN(\vvZ') &\equiv \ddt
        \semS{\ttG}{\flC(t)} \Big|_{t=0},\\
      \sN(v)
        &\equiv \sO(v) \text{ for } v \not\in \vvZ \cup \vvZp.
    \end{align*}
    To show that \(\sN \in \sem{\dap \rfle \dap{\vvX,\vvZ}{\fF \Land \vvZ = \ttG}}\), fix an arbitrary \(\vvX\)-regular trace \(\flA \in \semPS[\dap][\sN]\).
    It suffices to construct a \((\vvX,\vvZ)\)-regular trace \(\flB \in \flRegular[\vvX,\vvZ]\) such that \(\flB \in \semPS[\dap{\vvX,\vvZ}{\fF \Land \vvZ = \ttG}][\sN]\) and \(\coincide\).
    
    At any \(t \in \flInt\), the trace \(\flA\) satisfies the algebraic condition of Lemma \ref{lem:c_flow}.
    Since the \(C^2\) trajectory \(\flC\) is locally unique, \(\flA\) and \(\flC\) must be identical on a neighborhood \((t-\varepsilon,t+\varepsilon)\), ensuring that \(t \mapsto \flAt(\vvX)\) is \(C^2\) on \(\flInt\).
    Therefore, by the chain rule (Lemma \ref{lem:diff}) the semantic valuation \(t \mapsto \semS{\ttG}{\flAt}\) is of class \(C^1(\flInt,\R^m)\).

    We define the candidate trace \(\flB \in \flSpace\) by
    \begin{align*}
      &\flBt(\vvZ) \equiv \semS{\ttG}{\flAt}, \\
      &\flBt(\vvZp) \equiv \ddt \semS{\ttG}{\flAt}, \\
      &\flBt(v) \equiv \flAt(v) \text{ for } v \not\in \vvZ \cup \vvZp.
    \end{align*}
    We show that $\flB$ satisfies the required trace conditions:
    \begin{enumerate}
      \item Differential constraint: By construction, we have \(\flB(\flInt[\flB]) \subseteq \sem{\vvZ = \ttG}\).
      Applying the Coincidence Lemma pointwise, by the syntactic side condition \(\vvZ,\vvZp \not\in\FV{\fF}\) and \(\flA(\flInt) \subseteq \sem{F}\), we have \(\flB(\flInt[\flB]) \subseteq \sem{F}\).
      Thus, \(\flB(\flInt[\flB]) \subseteq \sem{F} \cap \sem{\vvZ = \ttG} = \sem{F \Land \vvZ = \ttG}\).
      \item \((\vvX,\vvZ)\)-regularity: The regularity of \(t \mapsto \flBt(\vvX)\) follows directly from the regularity of \(t \mapsto \flAt(\vvX)\).
      Moreover, \(t \mapsto \flBt(\vvZ)\) is of class \(C^1(\flInt[\flB],\R^m)\) by the regularity of \(t \mapsto \semS{\ttG}{\flAt}\).
    \end{enumerate}
    Combining these properties, we have \(\flB \in \semPS[\dap{\vvX, \vvZ}{\fF \Land \vvZ = \ttG}][\sN]\).
  \end{itemize}
\end{proof}
\begin{proof}[Soundness of Subsystem Axioms]
  \begin{itemize}
    \item[\lref{ax:dm}] Let \(\sO \in \sem{\dap \rfle \dap{\vvX}{\fG}}\).
    To show \(\sO \in \sem{\dap{\vvX}{\fF \Land \fR} \rfle \dap{\vvX}{\fG \Land \fR}}\), fix an arbitrary \(\vvX\)-regular trace \(\flA \in \semPS[\dap{\vvX}{\fF \Land \fR}]\).
    It suffices to show that there exists a trace \(\flB \in \semPS[\dap{\vvX}{\fG \Land \fR}]\) such that \(\coincide\).

    By the semantics of DAPs, \(\flA(\flInt) \subseteq \sem{\fF \Land \fR}  \subseteq \sem{\fF}\), which implies \(\flA \in \semPS[\dap]\).
    From the refinement assumption, it follows that there exists a trace \(\flB \in \semPS[\dap{\vvX}{\fG}]\) with \(\coincide\).

    Moreover, because \(\flA\) and \(\flB\) are \(\vvX\)-regular and share the initial state \(\sO\), they must remain constant outside \(\vvX \cup \vvXp\).
    Since they also coincidence on \(\vvX \cup \vvXp\) and have equal duration, they are identical.
    By \(\flA(\flInt) \subseteq \sem{\fR}\), we conclude \(\flB([0,T_\flB]) \subseteq \sem{\fR}\).
    Thus, \(\flB \in \semPS[\dap{\vvX}{\fG \Land \fR}]\), as required.
  \item[\lref{ax:dp}] 
    Let $\sO \in \States$ be such that
    \begin{align*}
      \sO &\in \sem{\fBox[\dap{\vvX}{\exists \vvYp \exists \vvY \fF}][(\fF)]}, \\
      \sO &\in \sem{\fBox[\dap{\vvX, \vvY}{\exists \vvYp \exists \vvY \fG}][(\fG)]},\\
      \sO &\in \sem{\dap \rfle \dap{\vvX}{\fG}}.
    \end{align*}
    Fix an arbitrary trace $\flA \in \semPS[\dap{\vvX,\vvY}{\fF}][\sO]$.
    It suffices to show that there exists a $(\vvX,\vvY)$-regular trace \(\flB \in \semPS[\dap{\vvX,\vvY}{\fG}]\) such that \( \coincide[\vvX,\vvY]\).
    
    We construct the trace $\tilde{\flA} \in \flSpace$ with duration $\flTB = \flTA$:
    \begin{align*}
    \tilde{\flA}(t)(\vvY) &\equiv \sO(\vvY), \\
      \tilde{\flA}(t)(\vvYp) &\equiv \sO(\vvYp), \\
      \tilde{\flA}(t)(v) &\equiv \flAt(v) \text{ for } v \not\in \vvY \cup \vvYp.
    \end{align*}
    This construction implies \(\tilde{\flA}(\flInt) \subseteq \sem{\exists \vvYp \exists \vvY\fF}\).
    By the assumption \(\sO \in \sem{\fBox[\dap{\vvX}{\exists \vvYp \exists \vvY\fF}][\fF]}\), we have \(\tilde{\flA}(\flInt) \subseteq \sem{\fF}\).
    Thus, \(\tilde{\flA} \in \sem{\dap}\).
    Applying the refinement assumption to $\tilde{\flA}$, there exists a trace \(\flB \in \semPS[\dap{\vvX}{\fG}]\) such that  \(\coincide[\vvX][\tilde{\flA}][\flB]\).

    We construct the trace $\tilde{\flB} \in \flSpace$ with duration $\flTB = \flTA$:
    \begin{align*}
    \tilde{\flB}(t)(\vvY) &\equiv \flAt(\vvY), \\
      \tilde{\flB}(t)(\vvYp) &\equiv \flAt(\vvYp), \\
      \tilde{\flB}(t)(v) &\equiv \flBt(v) \text{ for } v \not\in \vvY \cup \vvYp.
    \end{align*}
    By construction, we have $\tilde{\flB} \in \flRegular[\vvX,\vvY]$ and \(\tilde{\flB}(\flInt) \subseteq \sem{\exists \vvYp \exists \vvY G}\).
    Moreover, by assumption \(\sO \in \sem{\fBox[\dap{\vvX, \vvY}{\exists \vvYp \exists \vvY \fG}][\fG]}\), we have \(\tilde{\flB}(\flInt[\flB]) \subseteq \sem{G}\), completing the proof.\qed
  \end{itemize}
\end{proof}

\subsection{Derived Rules.}
\label{sec:appendex_derived_rules}
This section completes the proof of Theorem \ref{thm:derived_axioms}.
\begin{proof}
  We derive each of the rules individually:
  \begin{enumerate}
    \item [\lref{ax:jacobian}] The Jacobian axiom follows directly from the definition of syntactic (partial) derivatives (Definitions \ref{def:partial_diff}, \ref{def:jacobian}), note that as all definitions are syntactic, direct symbolic manipulations suffice.
    \item[\lref{rl:dw}] First, recall that the \lref{ax:k}~axiom
      \begin{align*}
        \axtag{ax:k} & \quad \lquote{ax:k}
      \end{align*}
      and the \lref{rl:g}~rule:
      \begin{sequentproof}[align]
        \ax[rl:g]{}{\fP}
        \un[rl:g]{\Gamma}{\fBox}
      \end{sequentproof}
      Next, applying the \lref{ax:k} and the \lref{ax:dw}~axiom, we have
      \begin{sequentproof}[align]
        \ax{\Gamma}{\fBox[\dap][(\fF \Limplies \fP)]}
        \ax*[ax:dw]{\Gamma}{\fBox[\dap][\fF]}
        \bi[ax:k]{\Gamma}{\fBox[\dap]}
      \end{sequentproof}
      The open goal can we further reduced by generalization:
      \begin{sequentproof}[align]
        \ax{\fF}{\fP}
        \un[rl:g]{\Gamma}{\fBox[\dap][(\fF \Limplies \fP)]}
      \end{sequentproof}

      Finally, composing the above derivation construct the desired derivation:
      \begin{sequentproof}[align]
        \ax{\fF}{\fP}
        \un[rl:dw]{\Gamma}{\fBox[\dap],\Delta}
      \end{sequentproof}
   \item[\lref{rl:da}] First, by transitivity we split with the intermediate formula $\fR$: 
      \begin{sequentproof}[align]
        \ax{\Gamma}{\dap{\vvX}{\fF} \rfle \dap{\vvX}{\fR}, \Delta}
        \ax{\Gamma}{\dap{\vvX}{\fR} \rfle \dap{\vvX}{\fG}, \Delta}
        \bi[ax:tr]{\Gamma}{\dap \rfle \dap{\vvX}{\fG}, \Delta}
      \end{sequentproof}
      Then, by another application of transitivity, we split the remaining goal via the intermediate constraint $\fF \Land \fR$:
      \begin{sequentproof}[align]
        \ax{\Gamma}{\dap{\vvX}{\fF} \rfle \dap{\vvX}{\fF \Land \fR}, \Delta}
        \ax*[ax:dr]{\Gamma}{\dap{\vvX}{\fF \Land \fR } \rfle \dap{\vvX}{\fR}, \Delta}
        \bi[ax:tr]{\Gamma}{\dap{\vvX}{\fF} \rfle \dap{\vvX}{\fR}, \Delta}
      \end{sequentproof}
      Finally, by differential cut and differential weakening, the refinement reduces to an implication between the corresponding constraints:
      \begin{sequentproof}[align]
        \ax[rl:dw]{\fF}{\fR}
        \un[rl:dw]{\Gamma}{\fBox[\dap][(\fF \Land \fR)],\Delta}
        \un[ax:dc]{\Gamma}{\dap{\vvX}{\fF} \rfle \dap{\vvX}{\fF \Land \fR}, \Delta}
      \end{sequentproof}
      Combining the above derivations we have the desired goal:
      \begin{sequentproof}
        \ax{\fF}{\fR}
        \ax{\Gamma}{\dap{\vvX}{\fR} \rfle \dap{\vvX}{\fG}, \Delta}
        \bi[rl:da]{\Gamma}{\dap \rfle \dap{\vvX}{\fG}, \Delta}
      \end{sequentproof}
    \item[\lref{rl:di}] 
      Unfolding the refinement equality, it suffices to show that
    \[
        \Gamma \vdash \dap{\vvX}{\fF \Land \ttE \preceq \ttZero} \rfle \dap
    \]
    and
    \[
      \Gamma \vdash\dap \rfle \dap{\vvX}{\fF \Land \ttE \preceq \ttZero}.
    \]
    We first consider the right-to-left refinement, which follows immediately from the differential refinement axiom~\lref{ax:dr}:
    \begin{sequentproof}
      \ax*[ax:dr]{\Gamma}{\dap{\vvX}{\ttE \preceq \ttZero \Land \fF} \rfle \dap, \Delta}
      \un[rl:da]{\Gamma}{\dap{\vvX}{\fF \Land \ttE \preceq \ttZero} \rfle \dap, \Delta}
    \end{sequentproof}

    To obtain the left-to-right refinement
    \[
      \Gamma \vdash \dap{\vvX}{\fF} \rfle \dap{\vvX}{\fF \Land \ttE \preceq \ttZero}, \Delta
     \] 
     we apply the transitivity axiom~\lref{ax:tr} along the following chain of formulas:
     \begin{align*}
       \fG_1 &\equiv \fF,\\
       \fG_2 &\equiv \fF \Land \D{\ttE} \leq \ttZero,\\
       \fG_3 &\equiv \fF \Land \ttE \preceq \ttZero \Land \D{\ttE} \leq \ttZero, \\
       \fG_4 &\equiv \fF \Land \ttE \preceq \ttZero.
     \end{align*}
     \begin{enumerate}
       \item First, by the \lref{ax:dc}~axiom, the refinement is reduced to a box-modality:
     \begin{sequentproof}
      \ax{\Gamma}{\fBox[\dap{\vvX}{\fF}][\D{\ttE} \leq \ttZero], \Delta}
      \un[rl:dw]{\Gamma}{\fBox[\dap{\vvX}{\fF}][(F \Land \D{\ttE} \leq \ttZero)], \Delta}
      \un[ax:dc]{\Gamma}{\dap{\vvX}{\fF} \rfle \dap{\vvX}{\fF \Land \D{\ttE} \leq \ttZero}, \Delta}
      \un[rl:unfold]{\Gamma}{\dap{\vvX}{\fG_1} \rfle \dap{\vvX}{\fG_2}, \Delta}
    \end{sequentproof}
    \item Next in the refinement chain, we reduce the expression to a box modality by removing the system constraint $\fF$ using axiom~\lref{ax:dm}, after which axiom~\lref{ax:dc} applies.
   \begin{sequentproof}
      \ax{\Gamma}{\fBox[\dap{\vvX}{\D{\ttE} \leq \ttZero}][\ttE \preceq \ttZero], \Delta}
      \un[ax:dc]{\Gamma}{\dap{\vvX}{\D{\ttE} \leq \ttZero} \rfle \dap{\vvX}{\D{\ttE} \leq \ttZero \Land \ttE \preceq \ttZero}, \Delta}
      \un[ax:dm]{\Gamma}{\dap{\vvX}{\D{\ttE} \leq \ttZero \Land \fF} \rfle \dap{\vvX}{\D{\ttE} \leq \ttZero \Land \ttE \preceq \ttZero \Land \fF}, \Delta}
      \un[rl:da]{\Gamma}{\dap{\vvX}{\fF \Land \D{\ttE} \leq \ttZero} \rfle \dap{\vvX}{\fF \Land \ttE \preceq \ttZero \Land \D{\ttE} \leq \ttZero}, \Delta}
      \un[rl:unfold]{\Gamma}{\dap{\vvX}{\fG_2} \rfle \dap{\vvX}{\fG_3}, \Delta}
    \end{sequentproof} 
    By $\vvXp \notin \ttE$, the remaining goal is closed by cutting in $\ttE \preceq \ttZero$ and applying the \lref{ax:dileq}~axiom:
    \begin{sequentproof}
      \ax{\Gamma}{\ttE \preceq \ttZero, \Delta}
      \ax*[ax:dileq]{\Gamma, \ttE \preceq \ttZero}{\fBox[\dap{\vvX}{\D{\ttE} \leq \ttZero}][\ttE \preceq \ttZero], \Delta}
      \bi[rl:cut]{\Gamma}{\fBox[\dap{\vvX}{\D{\ttE} \leq \ttZero}][\ttE \preceq \ttZero], \Delta}
    \end{sequentproof}
    \item The last refinement is immediately closed by \lref{ax:dr}:
    \begin{sequentproof}
      \ax*[ax:dr]{\Gamma}{\dap{\vvX}{\fF \Land \ttE \preceq \ttZero \Land \D{\ttE} \leq \ttZero} \rfle \dap{\vvX}{\fF \Land \ttE \preceq \ttZero }, \Delta}
      \un[rl:unfold]{\Gamma}{\dap{\vvX}{\fG_3}\rfle \dap{\vvX}{\fG_4 }, \Delta}
    \end{sequentproof}
   
    Finally, the desired derivation follows by transitivity form the above derivations: 
    \begin{sequentproof}[align]
      \ax{\Gamma}{\ttE \preceq \ttZero, \Delta}
      \ax{\Gamma}{\fBox[\dap][\D{\ttE} \leq \ttZero], \Delta}
      \bi[rl:di]{\Gamma}{\dap \rfeq \dap{\vvX}{\fF \Land \ttE \preceq \ttZero}, \Delta}
    \end{sequentproof}
    \end{enumerate}
    \item[\lref{rl:dhc}] 
      Unfolding the refinement equality, it suffices to show that
    \[
      \Gamma \vdash \dap{\vvX}{\fF \Land \D{\ttE} = \ttZero} \rfle \dap
    \]
    and
    \[
      \Gamma \vdash\dap \rfle \dap{\vvX}{\fF \Land \D{\ttE} = \ttZero}.
    \]
    We first consider the right-to-left refinement, which follows immediately from the differential refinement axiom~\lref{ax:dr}:
    \begin{sequentproof}
      \ax*[ax:dr]{\Gamma}{\dap{\vvX}{\D{\ttE} = \ttZero \Land \fF} \rfle \dap, \Delta}
      \un[rl:da]{\Gamma}{\dap{\vvX}{\fF \Land \D{\ttE} = \ttZero} \rfle \dap, \Delta}
    \end{sequentproof}

    To obtain the left-to-right refinement
    \[
      \Gamma \vdash \dap{\vvX}{\fF} \rfle \dap{\vvX}{\fF \Land \D{\ttE} = \ttZero}, \Delta
     \] 
     we apply the transitivity axiom~\lref{ax:tr} along the following chain of formulas:
     \begin{align*}
       \fG_1 &\equiv \fF,\\
       \fG_2 &\equiv \fF \Land \ttE = \ttZero,\\
       \fG_3 &\equiv \fF \Land \D{\ttE} = 0 \Land \ttE = \ttZero, \\
       \fG_4 &\equiv \fF \Land  \D{\ttE} = 0.
     \end{align*}
     \begin{enumerate}
       \item First, by the \lref{ax:dc}~axiom, the refinement is reduced to a box-modality:
     \begin{sequentproof}
      \ax{\Gamma}{\fBox[\dap{\vvX}{\fF}][\ttE = \ttZero], \Delta}
      \un[ax:dc]{\Gamma}{\dap{\vvX}{\fF} \rfle \dap{\vvX}{\fF \Land \ttE = \ttZero}, \Delta}
      \un[rl:unfold]{\Gamma}{\dap{\vvX}{\fG_1} \rfle \dap{\vvX}{\fG_2}, \Delta}
    \end{sequentproof}
    \item Next in the refinement chain, we reduce the expression to a box modality by removing the system constraint $\fF$ using axiom~\lref{ax:dm}, after which axiom~\lref{ax:dc} applies.
   \begin{sequentproof}
     \ax{\Gamma}{\fBox[\dap{\vvX}{\ttE = \ttZero}][\D{\ttE} = \ttZero], \Delta}
      \un[ax:dc]{\Gamma}{\dap{\vvX}{\ttE = \ttZero} \rfle \dap{\vvX}{\ttE = \ttZero \Land \D{\ttE} = \ttZero}, \Delta}
      \un[ax:dm]{\Gamma}{\dap{\vvX}{\ttE = \ttZero \Land \fF} \rfle \dap{\vvX}{\ttE = \ttZero \Land \D{\ttE} = \ttZero \Land \fF}, \Delta}
      \un[rl:da]{\Gamma}{\dap{\vvX}{\fF \Land \ttE = \ttZero} \rfle \dap{\vvX}{\fF \Land \D{\ttE} = \ttZero \Land \ttE = \ttZero}, \Delta}
      \un[rl:unfold]{\Gamma}{\dap{\vvX}{\fG_2} \rfle \dap{\vvX}{\fG_3}, \Delta}
    \end{sequentproof} 
    By $\vvXp \notin \ttE$, the remaining goal is closed by cutting in $\ttE \preceq \ttZero$ and applying the \lref{ax:dhc}~axiom:
    \begin{sequentproof}
      \ax{\Gamma}{\D{\ttE} = \ttZero, \Delta}
      \ax*[ax:dhc]{\Gamma, \D{\ttE} = \ttZero}{\fBox[\dap{\vvX}{\ttE = \ttZero}][\D{\ttE} = \ttZero], \Delta}
      \bi[rl:cut]{\Gamma}{\fBox[\dap{\vvX}{\ttE = \ttZero}][\D{\ttE} = \ttZero], \Delta}
    \end{sequentproof}
    \item The last refinement is immediately closed by \lref{ax:dr}:
    \begin{sequentproof}
      \ax*[ax:dr]{\Gamma}{\dap{\vvX}{\fF \Land \D{\ttE} = \ttZero \Land \ttE = \ttZero} \rfle \dap{\vvX}{\fF \Land \D{\ttE} = \ttZero }, \Delta}
      \un[rl:unfold]{\Gamma}{\dap{\vvX}{\fG_3}\rfle \dap{\vvX}{\fG_4 }, \Delta}
    \end{sequentproof}
   
    Finally, the desired derivation follows by transitivity form the above derivations:
    \begin{sequentproof}[align]
      \ax{\Gamma}{\D{\ttE} = \ttZero, \Delta}
      \ax{\Gamma}{\fBox[\dap][\ttE = \ttZero], \Delta}
      \bi[rl:dhc]{\Gamma}{\dap \rfeq \dap{\vvX}{\fF \Land \D{\ttE} = \ttZero}, \Delta}
    \end{sequentproof}
    \end{enumerate} 
  \end{enumerate}
\end{proof}

\subsection{Index Reduction.} \label{sec:index_reduction_proof}
\begin{proof}
  To derive the refinement property \[\Gamma \vdash \dap{\vvX}{\ttF_m = \ttZero} \rfeq \dap{\vvX}{\ttF_0 = \ttZero},\]
  we proceed by induction over \(m\).
  \begin{description}
    \item[Case $m = 0$.] This case is immediate by \lref{rl:da}:
    \begin{sequentproof}[align]
      \ax*[rl:real]{}{\ttF_{0} = \ttZero \leftrightarrow \ttF_{0} = \ttZero}
      \un[rl:da]{\Gamma}{\dap{\vvX}{\ttF_{0} = \ttZero} \rfeq \dap{\vvX}{\ttF_0 = \ttZero}}
    \end{sequentproof} 
    \item[Case $m = k+1$.] For the induction step, assume that the following is a valid derivation of \dAL:
    \begin{sequentproof}
      \ax{\Gamma}{\bigwedge_{i=0}^{k-1} \D{\ttR^A_i} = 0}
      \ax{}{\bigwedge_{i=0}^{k-1} (\ttF_i = \ttZero \rightarrow \ttR_i^A = \ttZero)}
      \bi{\Gamma}{\dap{\vvX}{\ttF_{k} = \ttZero} \rfeq \dap{\vvX}{\ttF_0 = \ttZero}}
    \end{sequentproof}
    
    To derive the condition for $m = k+1$, it suffices to prove that
    \begin{sequentproof}
      \ax{\Gamma}{\D{\ttR^A_k} = 0}
      \ax{\ttF_k = \ttZero}{\ttR_k^A = 0}
      \bi{\Gamma}{\dap{\vvX}{\ttF_{k+1} = \ttZero} \rfeq \dap{\vvX}{\ttF_k = \ttZero}}
    \end{sequentproof}
    is a valid derivation of \dAL.

    The $\lref{rl:dhc}$ rule allows adding the constraint $\D{\ttG_k^A} = 0$ to the system, since it is invariant under $\ttG_k^A = 0$ and $\vvXp \not\in \ttG_k$: 
    \begin{sequentproof}
      \ax{\Gamma}{\D{\ttR_k^A} = 0}
      \ax{\Gamma}{\fBox[\dap{\vvX}{\ttF_k = \ttZero}][\ttR_k^A = \ttZero]}
      \bi[rl:dhc]{\Gamma}{\dap{\vvX}{\ttF_k = \ttZero \Land \D{\ttR_k^A} = \ttZero} \rfeq \dap{\vvX}{\ttF_k = \ttZero}}
      \un[rl:unfold]{\Gamma}{\dap{\vvX}{\ttF_{k+1} = \ttZero} \rfeq \dap{\vvX}{\ttF_k = \ttZero}}
    \end{sequentproof}
    Finally, the box goal can be reduced to:
    \begin{sequentproof}
      \ax{\ttF_k = \ttZero}{\ttR_k^A = \ttZero}
      \un[rl:dw]{\Gamma}{\fBox[\dap{\vvX}{\ttF_k = \ttZero}][\ttR_k^A = \ttZero]}
    \end{sequentproof}
    Combining the above derivation steps, we obtain
    \begin{sequentproof}
      \ax{\Gamma}{\D{\ttR_k^A} = \ttZero}
      \ax{\ttF_k = \ttZero}{\ttR^A_k = \ttZero}
      \bi{\Gamma}{\dap{\vvX}{\ttF_{k+1} = \ttZero} \rfeq \dap{\vvX}{\ttF_k = \ttZero}}
    \end{sequentproof}
    An application of rule~$\lref{rl:rtr}$, along with propositional transformations, yields:
    \begin{sequentproof}
      \ax{\Gamma}{\bigwedge_{i=0}^{k} \D{\ttR_i^A} = \ttZero}
      \ax{}{\bigwedge_{i=0}^k(\ttF_i = \ttZero \rightarrow \ttR_i^A = \ttZero)}
      \bi{\Gamma}{\dap{\vvX}{\ttF_{k+1} = \ttZero} \rfeq \dap{\vvX}{\ttF_0 = \ttZero}}
    \end{sequentproof}
    This establishes the desired result.
  \end{description}
\end{proof}
\end{document}